\documentclass[journal]{IEEEtran}
\usepackage{amsfonts}
\usepackage{amsmath}
\usepackage{amssymb}
\usepackage{revsymb}
\usepackage{graphicx}
\usepackage{hyperref}
\usepackage{color}
\providecommand{\U}[1]{\protect\rule{.1in}{.1in}}

\newtheorem{theorem}{Theorem}

\newtheorem{example}[theorem]{Example}

\newtheorem{lemma}[theorem]{Lemma}

\newtheorem{remark}[theorem]{Remark}

\newenvironment{proof}[1][Proof]{\noindent\textbf{#1.} }{\ \rule{0.5em}{0.5em}}

\newcommand{\bra}[1]{\left\langle #1 \right|}
\newcommand{\ket}[1]{\left|#1\right\rangle}

\newcommand{\tr}{\text{Tr}}

\newcommand{\bal}{\begin{align*}}
\newcommand{\eal}{\end{align*}}
\newcommand{\bpm}{\begin{pmatrix}}
\newcommand{\epm}{\end{pmatrix}}

\newcommand{\balign}{\begin{align*}}
\newcommand{\ealign}{\end{align*}}

\newcommand{\1}{\openone}
\def\cH{\mathcal{H}}
\def\cN{\mathcal{N}}

\def\cD{\mathcal{D}}
\def\cE{\mathcal{E}}

\newcommand{\be}{\begin{equation}}

\newcommand{\ee}{\end{equation}}

\newcommand{\eps}{{\varepsilon}}

\newcommand{\bea}{\begin{eqnarray}}
\newcommand{\eea}{\end{eqnarray}}

\DeclareMathOperator{\id}{id}
\def\reff#1{(\ref{#1})}
\DeclareRobustCommand\openone{\leavevmode\hbox{\small1\normalsize\kern-.33em1}}

\let\originalleft\left
\let\originalright\right
\def\left#1{\mathopen{}\originalleft#1}
\def\right#1{\originalright#1\mathclose{}}

\begin{document}

\title{Quantum rate distortion coding \protect\\ with auxiliary resources}

\author{Mark M. Wilde, Nilanjana Datta, Min-Hsiu Hsieh, and Andreas Winter\thanks{Mark M. Wilde is with the
School of Computer Science, McGill University, Montr\'{e}al, Qu\'{e}bec,
Canada H3A 2A7. Nilanjana Datta
is with the Statistical Laboratory, University of Cambridge,
Wilberforce Road, Cambridge CB3 0WB, United Kingdom.
Min-Hsiu Hsieh is with
Centre for Quantum Computation and Intelligent Systems (QCIS), Faculty
of Engineering and Information Technology (FEIT), University of
Technology, Sydney (UTS), PO Box 123, Broadway NSW 2007, Australia.
Andreas Winter is with ICREA \&{} F\'{\i}sica Te\`{o}rica: Informaci\'{o} i Fenomens Qu\`{a}ntics,Universitat Aut\`{o}noma de Barcelona, ES-08193 Bellaterra (Barcelona), Spain,
the Department of Mathematics, University of Bristol, Bristol BS8 1TW, UK,
and the Centre for Quantum Technologies, National University of Singapore, 2 Science Drive 3, Singapore 117542, Singapore.
}}

\maketitle

\begin{abstract}
We extend quantum rate distortion theory by considering auxiliary resources
that might be available to a sender and receiver performing lossy quantum data compression.
The first setting we consider
is that of quantum rate distortion coding with the help of a classical side channel. Our result
here is that the regularized entanglement of formation characterizes the quantum rate
distortion function, extending earlier work of Devetak and Berger.
We also combine this bound with the entanglement-assisted bound from our prior work
to obtain the best known bounds on the quantum rate distortion function for an isotropic
qubit source. The second setting we consider is that of quantum rate distortion coding with quantum side
information (QSI) available to the receiver. In order to prove results in this setting, we first state and prove
a quantum reverse Shannon theorem with QSI (for
tensor-power states), which extends the known
tensor-power quantum reverse Shannon theorem. The achievability part of this
theorem relies on the quantum state redistribution protocol, while the converse
relies on the fact that the protocol can cause
only a negligible disturbance to the joint state of the
reference and the receiver's QSI. This quantum reverse
Shannon theorem with QSI naturally leads to quantum rate-distortion theorems
with QSI, with or without entanglement assistance.

\end{abstract}


\section{Introduction}

Schumacher proved that the optimal rate of data compression of a memoryless,
quantum information source is given by its von Neumann entropy \cite{Schumacher:1995dg}. This data compression limit was evaluated under the requirement that the data compression scheme is {\em{lossless}}, in the sense that the information emitted by the source is recovered with arbitrary precision in the limit of asymptotically many copies of the source. However,
the lack of sufficient storage could make it necessary to compress a source beyond its von Neumann entropy. By the converse of Schumacher's theorem, this would mean that the information recovered after the compression-decompression scheme would suffer a certain amount of distortion compared to the original information. In other words, the data compression scheme would be {\em{lossy}}. 

The theory of lossy quantum data compression is called {\em{quantum rate distortion theory}}, in analogy with its classical counterpart developed by Shannon \cite{Shannon:tf}. It deals with the trade-off between the rate of compression and the allowed distortion. The trade-off is characterized by a rate distortion function which is defined as the minimum rate of data compression for a given distortion, with respect to a suitably defined distortion measure.

In the first paper to discuss quantum rate distortion theory, Barnum considered a symbol-wise entanglement fidelity as a distortion measure \cite{B00}. With respect to it, he obtained a lower bound on the quantum rate distortion function in terms of an entropic quantity, namely, the coherent information \cite{SN96}. Even though this was the first result in quantum rate distortion theory, it is unsatisfactory since the coherent information can become negative, whereas the rate distortion function, by its very definition, is always non-negative.

In \cite{DHW11}, we obtained an expression for the quantum rate distortion function in terms of the entanglement of purification \cite{THLD02}, which, in contrast to the coherent information, is always non-negative. However, our result too is not entirely satisfactory since the expression is given in terms of a regularized formula and hence cannot be effectively computed. Furthermore, there is recent evidence that the entanglement of purification is
a non-additive quantity \cite{CW12}, which if true would lead to further complications in evaluating the expression. The search for a single-letter formula for the quantum rate distortion function hence remains an important open problem.

It is often convenient to consider data compression in the communication paradigm, in which a sender (say, Alice) compresses the information emitted by the quantum information source and sends it to a receiver (say, Bob) who then decompresses it. In this setting, one considers Alice and Bob to have additional, auxiliary resources which they can employ to assist them in their compression-decompression task. One such auxiliary resource is prior shared entanglement between Alice and Bob. In \cite{DHW11}, we proved that in its presence, the quantum rate distortion function is characterized by a single-letter expression in terms of the quantum mutual information. This expression obviously provides a single-letter lower bound on the unassisted quantum rate distortion function, since, for any given distortion, the extra resource could in principle allow for improved compression. Furthermore, this result demonstrates that the coherent information (at least in the form suggested by Barnum \cite{B00}) is irrelevant for the task of unassisted quantum rate distortion coding because half the quantum mutual information is a lower bound on the unassisted quantum rate distortion function and it is also an upper bound on the coherent information.

\section{Summary of Results}

In this paper, we consider rate-distortion coding in the presence of other auxiliary resources, e.g., access to a noiseless classical side channel or quantum side information.
Doing so not only provides new bounds on the unassisted quantum rate distortion function, but it
also offers new scenarios that are unique to the quantum setting.

Alice and Bob are said to have a noiseless, forward classical side channel if Alice is allowed unlimited classical communication to Bob. Quantum rate distortion in the presence of such an auxiliary resource was studied for the special case of an isotropic qubit source -- i.e., one that produces a maximally mixed state on a qubit system -- by Devetak and Berger \cite{Devetak:2002it}. We consider the general case of an arbitrary, memoryless quantum information source, and prove that the corresponding rate distortion function is given in terms of a regularized entanglement of formation~\cite{BDSW96}. This classically-assisted rate distortion function serves as an alternate lower bound to the unassisted quantum rate distortion function, and we show that it can be tighter than the above entanglement-assisted lower bound in some cases.

Quantum rate distortion coding in the presence of quantum side information (QSI) corresponds to the following setting: Suppose a third party (say, Charlie) maps the source state $\rho$ via some isometry to a bipartite state $\rho_{AB}$, and distributes the systems $A$ and $B$ to Alice and Bob, respectively. The goal is for Alice to transfer her system $A$ to Bob, up to some given distortion, using as little quantum communication as possible. The rate distortion function is then defined as the minimum rate of quantum communication required for this task, evaluated in the limit in which Alice and Bob share asymptotically many copies of the state $\rho_{AB}$. We obtain an expression for the rate distortion function under the assumption that the protocol causes asymptotically negligible disturbance to the joint state of Bob and a reference system $R$ that purifies the state $\rho_{AB}$. This assumption may be motivated naturally in light of the fact that Bob may want to reuse his quantum side information in some future information processing task.
Furthermore, if we allow Alice and Bob to have sufficient prior shared entanglement in addition to Bob's QSI, then, under the above assumption, the rate distortion function is given by a single-letter expression in terms of the quantum conditional mutual information.
A classical analogue of the above problem was solved by Wyner and Ziv \cite{WZ76}.
Our results on quantum rate distortion with QSI also generalize Luo and Devetak's results on classical rate distortion in the presence of QSI \cite{LD09} and our prior results in \cite{DHWW12} on quantum-to-classical rate distortion coding with QSI.

The techniques used to prove these results are generalizations of those which we employed in \cite{DHW11} to obtain expressions for the unassisted and entanglement-assisted quantum rate distortion functions. The main ingredient in the proof of the achievability part for the entanglement-assisted case in \cite{DHW11} is the quantum reverse Shannon theorem \cite{BDHSW09,BCR09}, which quantifies the minimum rate of quantum communication required from Alice to Bob in order to asymptotically simulate a memoryless quantum channel, when they share entanglement. (It has been pointed out in both \cite{SV96} and \cite{W02} how a reverse Shannon theorem immediately leads to a rate distortion protocol.)
Analogously, to establish the achievability of our rate distortion functions in the presence of QSI, we employ a generalization of the quantum reverse Shannon theorem to the case in which Bob has QSI as an auxiliary resource. This theorem constitutes a result which is interesting in its own right, and the protocol of quantum state redistribution \cite{DY08,YD09} plays a key role in the proof of the achievability part. 
The achievability of the rate distortion function in the 
unassisted case was proved by using Schumacher compression~\cite{DHW11}, which can be viewed as a special reverse Shannon theorem where the goal is to simulate the identity channel. Analogously, the achievability of our rate distortion function in the presence of a classical side channel is proved by exploiting a variant of Schumacher compression.

The converse proofs of the results in this paper have certain 
similarities, using various identities and entropic inequalities, e.g., the quantum-data processing inequality \cite{SN96}, superadditivity of the quantum mutual information, and the Alicki-Fannes inequality \cite{AF04}. However, a non-trivial aspect of the QSI converse proofs is that they exploit the assumption that the protocol causes only a negligible disturbance to the joint state of the reference system and Bob's QSI. We can then invoke Uhlmann's theorem \cite{U73} as in Refs.~\cite{WHBH12,DHWW12} to demonstrate the existence of a map which does not act at all on Bob's quantum side information, such that the overall map on the source state is close to the original map that acts in part on Bob's quantum side information. As such, this feature of the converse proof is unique to quantum information theory and simply is not present in related classical results \cite{WZ76}.

The paper is organized as follows. In Section~\ref{prelim}, we introduce necessary notation and definitions, especially for the entropic quantities arising in the statements of the theorems. In Section~\ref{quantum}, we introduce the basic concepts of quantum rate distortion theory, and unify our results from \cite{DHW11} on unassisted and entanglement-assisted quantum rate distortion functions in order to obtain a trade-off between the rate of compression and the rate of entanglement consumption. We analyse quantum rate distortion in the presence of a classical side channel in Section~\ref{classical} and in the presence of QSI in 
Section~\ref{sec-qrd-qsi}. The latter section also contains our theorem on the rate distortion function in the presence of both QSI and prior shared entanglement. The necessary generalizations of the quantum reverse Shannon theorem required for our proofs are stated and proved in Section~\ref{qrst}. We conclude in Section~\ref{sec:conclusion} with a summary and some open questions.

\section{Notations and Definitions}
\label{prelim}
\label{sec-def}
Let $\mathcal{B}(\mathcal{H})$ denote the algebra of linear
operators acting on a finite-dimensional Hilbert space $\mathcal{H}$, and let
$\mathcal{D}(\mathcal{H})$ denote the set of positive operators of unit trace
(density operators) acting on $\mathcal{H}$. For any given pure state $|\psi\rangle
\in\mathcal{H}$ we denote the projector $|\psi\rangle\langle\psi|$ simply as
$\psi$. The trace distance between two operators $M$ and $N$ is given by
$$\left\Vert {M-N}\right\Vert_{1}\equiv\text{Tr}(|M-N|),$$ where $|C|\equiv
\sqrt{C^{\dag}C}$. Throughout this paper we restrict our considerations to
finite-dimensional Hilbert spaces, and we take the logarithm to base two.

We denote the Hilbert space associated
to a quantum system~$A$ by $\cH_A$, and the quantum systems
corresponding to $n$ copies of a pure state $\psi_{ABC}^{\otimes n}$
by $A^n$, $B^n$ and $C^n$. For a multiparty state~$\rho_{AB}$, we unambiguously refer to its reduced states
on systems $A$ and $B$ simply as $\rho_A$ and $\rho_B$, respectively, so that it is implicit that
$\rho_A = \tr_B(\rho_{AB})$ and $\rho_B = \tr_A(\rho_{AB})$.
Moreover, we denote a completely positive trace-preserving (CPTP) map
$\mathcal{N}:\mathcal{B}(\mathcal{H}_{A})\rightarrow\mathcal{B}(\mathcal{H}%
_{B})$ simply as $\mathcal{N}_{A\rightarrow B}$. Similarly we denote an
isometry $U:\mathcal{B}(\mathcal{H}_{A})\rightarrow\mathcal{B}(\mathcal{H}%
_{B}\otimes\mathcal{H}_{E})$ simply as $U_{A\rightarrow BE}$. The identity map
on states in $\mathcal{D}(\mathcal{H}_{A})$ is denoted as $\mathrm{{id}}_{A}$.

The von Neumann entropy of a state $\rho\in\mathcal{D}(\mathcal{H}_{A})$ is
defined as $$H(\rho)\equiv H(A)_\rho \equiv -\text{Tr}(\rho\log\rho),$$ and satisfies
the bound $H(\rho) \le \log{\rm{dim }} (\cH_A)$. The conditional entropy
$H(A|B)_{\rho}$,  the quantum mutual information $I(A;B)_{\rho}$, and the conditional quantum mutual
information $I(A;B|C)_{\rho}$ of a tripartite state $\rho_{ABC}$ are defined as follows:
\bea
H(A|B)_{\rho} & \equiv & H(AB)_\rho - H(B)_\rho, \nonumber\\
I(A;B)_{\rho} & \equiv & H(A)_\rho - H(A|B)_{\rho}, \nonumber\\
I(A;B|C)_{\rho} & \equiv & H(A|C)_\rho - H(A|BC)_{\rho}.
\eea
The entanglement of formation \cite{BDSW96} of a bipartite state $\rho_{AB}$ is 
defined as
\be\label{eof} 
E_F(\rho_{AB}) \equiv \min_{\{p(x), |\psi^x_{AB}\rangle\}} \sum_x p(x) \, H(A)_{\psi^x},
\ee
where $\{p(x), |\psi^x_{AB}\rangle\}$ is an ensemble of pure states such that
$\rho_{AB} = \sum_x p(x) \, |\psi^x_{AB}\rangle\langle \psi^x_{AB}|$. 
The entanglement of purification \cite{THLD02} of a bipartite state $\rho_{AB}$ is 
given by the following expression:%
\be\label{eop}
E_p(\rho_{AB}) = \min_{\cN} H\left(({\rm{id}}_B \otimes \cN_{E\to E^\prime})(\sigma_{BE}(\rho)) \right),
\ee
where $\sigma_{BE}(\rho)= \tr_A (\phi^\rho_{ABE})$, $ \phi^\rho_{ABE}$ is some purification of $\rho_{AB}$, and the minimization is over all CPTP maps $\cN_{E\to E^\prime}$ acting on the system $E$.

We also employ the following lemmas in our proofs.

\begin{lemma}
[Quantum data processing inequality \cite{SN96}]\label{dataproc} If \
$\omega_{AB^{\prime}} = (\mathrm{{id}}_{A} \otimes\mathcal{N}_{B\to B^{\prime
}}) (\sigma_{AB})$, where $\mathcal{N}_{B\to B^{\prime}}$ is a CPTP map, then
\begin{equation}
I(A;B)_{\sigma}\ge I(A;B^{\prime})_{\omega}.
\end{equation}

\end{lemma}

\begin{lemma}
[Alicki-Fannes Inequality \cite{AF04}]\label{Thm:AF}Suppose two
states $\rho_{AB}$ and $\sigma_{AB}$ are close in trace distance:%
\[
\left\Vert \rho_{AB}-\sigma_{AB}\right\Vert _{1}\leq\eps, \quad {\hbox{for some }} \eps \ge 0.
\]
Then their respective conditional entropies are close:%
\begin{equation}
\left\vert H(A| B)_{\rho}-H(A| B)_{\sigma}\right\vert
\leq4\eps\log\left\vert A\right\vert +2 \, h_{2}(\eps),
\end{equation}
where $h_2(\eps)\equiv - \eps \log \eps - (1-\eps) \log (1-\eps)$ is the binary entropy.
\end{lemma}

In our converse proofs, we often deal with information quantities evaluated on $n$ systems. For simplicity,
rather than listing the precise bound from the Alicki-Fannes inequality,
we often just state an upper bound as $n \varepsilon^\prime$, where it is implicit that
$\varepsilon^\prime =  c_1\epsilon\log d + c_2 \, h_{2}(\epsilon)/n$ for positive constants
$c_1$ and $c_2$, and $d$ being the dimension of a single system.

Note that it is easy to derive continuity inequalities for the
quantum mutual information and the conditional quantum mutual information from the Alicki-Fannes inequality
simply by employing the triangle inequality.

\begin{lemma}
[Superadditivity of q.~mutual information \cite{DHW11}]\label{super}
The quantum mutual information is superadditive in the sense that, for any CPTP map
$\mathcal{N}_{A_{1}A_{2} \to B_{1}B_{2}}$,%
\[
I\left(  R_{1}R_{2};B_{1}B_{2}\right)  _{\sigma}\geq I\left(  R_{1}%
;B_{1}\right)  _{\sigma}+I\left(  R_{2};B_{2}\right)  _{\sigma},
\]
where
\[
\sigma_{R_{1}R_{2}B_{1}B_{2}} = \mathcal{N}_{A_{1}A_{2} \to B_{1}B_{2}}
\left(  \phi_{R_{1}A_{1}} \otimes\varphi_{R_{2}A_{2}} \right)  ,
\]
and $\phi_{R_{1}A_{1}}$ and $\varphi_{R_{2}A_{2}}$ are pure bipartite states.
\end{lemma}

We make frequent use of the following lemma in our converse proofs. For completeness, we give its proof.

\begin{lemma}
[Duality of Conditional Quantum Entropy]\label{dual_cond_ent}

For a tripartite pure state
$\psi_{ABC}$,%
\[
H\left(  A|B\right)  _{\psi}=-H\left(  A|C\right)  _{\psi}.
\]
\end{lemma}

\begin{proof}
\begin{align*}
H(A|B)_{\psi} & = H(AB)_{\psi} - H(B)_{\psi} \\
& = H(C)_{\psi} - H(AC)_{\psi} \\
& = -H\left(  A|C\right)  _{\psi}.
\end{align*}
\end{proof}

%

\section{Quantum Rate Distortion}
\label{quantum}

Throughout this paper we consider a memoryless quantum information source
characterized by a density matrix $\rho \in \cD(\cH_A)$. We refer to $\rho$
as the {\em{source state}} and denote a purification of it by $\psi_{RA}^\rho
=  |\psi^\rho_{RA}\rangle \langle \psi^\rho_{RA}|$, with $R$ being a purifying
reference system isomorphic to $A$. 
The state $\rho^{n}\equiv \rho^{\otimes n} \in \mathcal{D}(\mathcal{H}%
_{A}^{\otimes n})$ characterizes $n$ copies of the source. 
A \emph{source coding} (or \emph{compression-decompression})
scheme of rate~$R$ is defined by a block code,\footnote{It should be very clear 
from the context whether $R$ refers to ``rate'' or
``reference system.''} which consists of two quantum
operations---the encoding and decoding maps. The encoding $\mathcal{E}_{n}$ is
a CPTP map from $n$ copies of the source space to a Hilbert space ${\widetilde
{\mathcal{H}}_{Q^{n}}}$ of dimension $\approx 2^{nR}$:%
\[
\mathcal{E}_{n}:\mathcal{D}(\mathcal{H}_{A}^{\otimes n})\rightarrow
\mathcal{D}({\widetilde{\mathcal{H}}_{Q^{n}}}),
\]
and the decoding $\mathcal{D}_{n}$ is a CPTP map from the compressed space to 
the original Hilbert space $\mathcal{H}_{A}^{\otimes n}$:%
\[
\mathcal{D}_{n}:\mathcal{D}({\widetilde{\mathcal{H}}_{Q^{n}}})\rightarrow
\mathcal{D}(\mathcal{H}_{A}^{\otimes n}).
\]
The average distortion resulting from this compression-decompression scheme is
defined as an average \cite{Shannon:tf,B00}:%
\[
{\overline{d}}(\rho,\mathcal{D}_{n}\circ\mathcal{E}_{n})
   \equiv \frac{1}{n}\sum_{i=1}^{n}d(\rho,\mathcal{F}_{n}^{(i)}),
\]
where $\mathcal{F}_{n}^{(i)}$ is the \textquotedblleft marginal
operation\textquotedblright\ on the $i$-th copy of the source space induced by
the overall operation $\mathcal{F}_{n}\equiv\mathcal{D}_{n}\circ
\mathcal{E}_{n}$, and is defined as%
\begin{equation}
\label{fni}
  \mathcal{F}_{n}^{(i)}(\xi) 
     \equiv \mathrm{{Tr}}_{A_{1}A_{2}\cdots A_{i-1}A_{i+1}\cdots A_{n}}
                  [\mathcal{F}_{n}(\rho^{\otimes(i-1)} \otimes \xi \otimes \rho^{\otimes(n-i)})],
\end{equation}
and for any CPTP map $\cN$,
\[
d(\rho,\mathcal{N})=1-F_{e}(\rho,\mathcal{N}),
\]
with $F_{e}$ being the entanglement fidelity of $\mathcal{N}$:%
\begin{equation}
F_{e}(\rho,\mathcal{N})\equiv\langle\psi^{\rho}_{RA}|({\mathrm{{id}}}%
_{R}\otimes\mathcal{N}_{A\rightarrow A})(\psi_{RA}^{\rho})|\psi^{\rho
}_{RA}\rangle. \label{fid1}%
\end{equation}
The quantum operations $\mathcal{D}_{n}$ and $\mathcal{E}_{n}$ define an
$(n,R)$ quantum rate distortion code.

Motivated by the observation that $\overline{d}$ is the average of a linear
function of the marginal channels, and by Shannon's rate distortion theory
\cite{Shannon:tf}, which allows a general average-type distortion function
of input and output, we can generalize the above setting as follows
\cite{WA01,CW08,DHWW12}:
Let $\Delta \geq 0$ be an observable on $RB$, which we will (without loss of generality)
assume to be non-negative. This distortion observable~$\Delta$
will define the distortion function between
a reference for the input ($R$) and the output ($B$) of the compression and decompression maps
$\mathcal{E}_{n}:\mathcal{D}(\mathcal{H}_{A}^{\otimes n})
 \rightarrow \mathcal{D}({\widetilde{\mathcal{H}}_{Q^{n}}})$ and
$\mathcal{D}_{n}:\mathcal{D}({\widetilde{\mathcal{H}}_{Q^{n}}})
 \rightarrow \mathcal{D}(\mathcal{H}_{B}^{\otimes n})$, respectively,
with $R$ not necessarily equal to $B$. With this, we can define the distortion
of a channel $\mathcal{N}$ with respect to a source $\rho$, as
\[
  d(\rho,\mathcal{N}) 
  \equiv \mathrm{Tr}\left( \Delta_{RB} \left( (\id_R \otimes \mathcal{N}_{A\rightarrow B})(\psi_{RA}^\rho) \right)\right),
\]
and the average distortion on a block of $n$ as
\[\begin{split}
  {\overline{d}}(\rho,\mathcal{D}_{n}\circ\mathcal{E}_{n})
      &\equiv \frac{1}{n}\sum_{i=1}^{n}d(\rho,\mathcal{F}_{n}^{(i)}) \\
      &=      \mathrm{Tr}\left( {\overline{\Delta}} 
                      \left( (\id_R \otimes \mathcal{F}_n)\left((\psi_{RA}^\rho)^{\otimes n}\right) \right)\right).
\end{split}\]
Here, the marginal channels $\mathcal{F}_{n}^{(i)}$ are now given by
\[
  \mathcal{F}_{n}^{(i)}(\xi) 
   \equiv \mathrm{{Tr}}_{B_{1}B_{2}\cdots B_{i-1}B_{i+1}\cdots B_{n}}
            [\mathcal{F}_{n}(\rho^{\otimes(i-1)} \otimes \xi \otimes \rho^{\otimes(n-i)})],
\]
and
\[
  \overline{\Delta} 
     \equiv \frac{1}{n}\sum_{i=1}^n \1^{\otimes(i-1)} \otimes \Delta \otimes \1^{\otimes(n-i)}
\]
is the average distortion observable.

Allowing different input and output spaces, and comparing them via the
otherwise completely arbitrary observable $\Delta$ may look like a drastic 
departure from the source coding paradigm, but looking at Shannon's
rate distortion theory \cite{Shannon:tf} and its applications reveals that
it is indeed very natural -- see also the examples below.

For any $R,D\geq0$, the pair $(R,D)$ is said to be an \emph{achievable} rate
distortion pair if there exists a sequence of $(n,R)$ quantum rate distortion
codes $(\mathcal{E}_{n},\mathcal{D}_{n})$ such that%
\begin{equation}
\label{avg_dist}\lim_{n\rightarrow\infty}{\overline{d}}(\rho,\mathcal{D}%
_{n}\circ\mathcal{E}_{n})\leq D.
\end{equation}
The \emph{{quantum rate distortion function}} is then defined as
\[
R^{q}(D)=\inf\{R:(R,D) \text{ is achievable}\}.
\]
In \cite{DHW11}, we proved that the quantum rate distortion function admits the following characterization
in terms of the regularized entanglement of purification:
\begin{equation}
R^{q}\left(  D\right)  =\lim_{k\rightarrow\infty}\frac{1}{k}\left[
\min_{\mathcal{N}^{\left(  k\right)  }\ :\ \overline{d}(\rho,\mathcal{N}%
^{\left(  k\right)  })\leq D}\ E_{p}\left(  \rho^{\otimes k},\mathcal{N}%
^{\left(  k\right)  }\right)  \right]  . \label{eq:qrd-dhw11}
\end{equation}

In the presence of an auxiliary resource, the rate distortion function is defined 
analogously, the only difference being in the encoding and decoding maps.
In particular, if Alice and Bob have 
prior shared entanglement,
then, denoting the entangled systems by ${T_A}$ and ${T_B}$ 
(with $T_A$ being with Alice and $T_B$ being with Bob), the encoding
and decoding maps are respectively given by
\be\label{ea-enc}
\mathcal{E}_{n}:\mathcal{D}(\mathcal{H}_{A}^{\otimes n}\otimes \cH_{T_A})\rightarrow
\mathcal{D}({\widetilde{\mathcal{H}}_{Q^{n}}}),
\ee
and
\be
\label{ea-dec}
\mathcal{D}_{n} : \mathcal{D}({\widetilde{\mathcal{H}}_{Q^{n}}} \otimes \cH_{T_B} )\rightarrow
\mathcal{D}(\mathcal{H}_{B}^{\otimes n}).
\ee
We denote the corresponding entanglement-assisted quantum rate distortion function, 
for a given distortion $D \ge 0$ and unlimited amount of entanglement, 
as $R_{ea}^q(D)$. In \cite{DHW11}, we proved that the
entanglement-assisted quantum rate distortion function is equal to the following
single-letter expression:
$$
R^q_{ea}(D) = \frac{1}{2} \left[\min_{\mathcal{N} \ : \ d(\rho,\mathcal{N})\leq D}
\ I\left( R;B\right)_\omega  \right]
$$
where $\omega_{RB} = \mathcal{N}_{A \to B} (\psi^{\rho}_{RA})$.

\begin{remark}
\label{rem:convexity}
It should be noted that for every choice of distortion observable $\Delta$,
the quantum rate distortion function $R^q(D)$ is convex in the distortion $D$, and likewise 
the entanglement-assisted rate distortion function $R^q_{ea}(D)$. 
This is seen by observing that time-sharing codes of rate $R_j$ and
distortion $D_j$ ($j=1,2$) on fractions $\lambda$ and $1-\lambda$ of a block,
yields directly a code of rate $\lambda R_1+(1-\lambda) R_2$ and
distortion $\lambda D_1+(1-\lambda) D_2$.
\end{remark}

By the above observation and the coding theorem expressed
by (\ref{eq:qrd-dhw11}), we can conclude that the regularized expression on the RHS of
(\ref{eq:qrd-dhw11}) is convex in $D$.
This is true, even though the convexity of the expression on
the RHS of (\ref{eq:qrd-dhw11}) is not  immediately evident,
since it is known that the entanglement of purification and even its regularization are not convex in the state
on which it is being evaluated \cite{THLD02}.
In fact, all the coding theorems in this paper contain expressions
for the rate distortion function
(or for rate regions) that are convex in the distortion parameter $D$.
Indeed, one well-known way of proving convexity of an expression for
a rate-distortion function is as outlined in Lemma~14
of \cite{DHW11}, and this approach relies on convexity of the underlying
information measure with respect to a distortion channel. Thus, for any
finite $k$, the expression in (\ref{eq:qrd-dhw11}) is {\it not} convex in $D$
because the entanglement of purification is not convex, and it is only
in the regularized limit that this expression is convex in $D$. For the curious
reader, we show in Appendix~\ref{app:convexity-QRD} that the mathematical expression
on the RHS of (\ref{eq:qrd-dhw11}) is convex in~$D$.

\begin{example}
{\rm The original distortion measure based on entanglement fidelity is recovered
in the case where $A=B$, by letting the distortion observable $\Delta = \1-\psi^\rho$.}
\end{example}

\begin{example}
\label{ex:classical-rd}
{\rm Given a classical distortion function 
$d:\mathcal{X}\times\mathcal{Y} \rightarrow \mathbb{R}_{\geq 0}$, as considered 
in \cite{Shannon:tf}, for input and output alphabets $\mathcal{X}$ and $\mathcal{Y}$,
we consider $\mathcal{H}_A = \mathbb{C}^{\mathcal{X}}$ 
and $\mathcal{H}_B = \mathbb{C}^{\mathcal{Y}}$, and let
\[
  \Delta \equiv \sum_{xy} d(x,y) \ket{x}\bra{x} \otimes \ket{y}\bra{y}.
\]
In classical rate distortion theory we also consider an IID source with
single-letter marginal probability distribution $P(x)$, giving rise to the 
diagonal source density $\rho = \sum_x P(x) \ket{x}\bra{x}$ and its purification
$\ket{\psi_\rho} = \sum_x \sqrt{P(x)} \ket{x}\ket{x}$.

Now, a classical source coding scheme of rate $R$ and distortion $D$
naturally turns into a source code in the above quantum sense,
by lifting the stochastic encoding and decoding maps to CPTP maps
sending diagonal matrices to diagonal matrices; 
furthermore the quantum version still has rate $R$ (now qubits) and
the same distortion $D$. 

Conversely, given a source code in our above
sense, the relation to classical Shannon-style rate distortion coding
is a little more subtle. To start, however, we can at least say that
without loss of generality the overall map $\mathcal{F}_n$ is a classical channel from
$\mathcal{X}^n$ to $\mathcal{Y}^n$, because we can dephase the input to
$\mathcal{E}_n$ in the $x$-basis, and the output of $\mathcal{D}_n$
in the $y$-basis, without affecting rate or distortion. The compressed
system $Q^n$ may of course still use quantum states in a nontrivial way,
for instance if unlimited entanglement is available, to ``superdense-code'' \cite{PhysRevLett.69.2881}
the classical compressed information of a classical rate distortion code, 
thus halving the bit rate to the qubit rate. 
On the other hand, Theorem~3 of \cite{DHW11} shows that this is the only
improvement; indeed, the entanglement-assisted rate distortion function
is exactly half the classical one $R(D)$ in \cite{Shannon:tf}.
From this we can deduce that the unassisted quantum rate distortion 
function equals the classical one.
If we assume the opposite,
namely that the former were strictly smaller than the latter, then by invoking remote state
preparation of the compressed quantum states at an asymptotic cbit/qubit
rate of one \cite{BHLSW05}, and then ``superdense-coding'' the classical bits, 
we would get an entanglement-assisted rate distortion code of the same
distortion but rate $\tfrac{1}{2}R^q(D) < \tfrac{1}{2}R(D)$, resulting
in a contradiction.}
\end{example}

\begin{example}
\label{ex:q-c-rd}
{\rm Quantum-to-classical rate distortion \cite{DHWW12} is based on an
observable of the form
\[
  \Delta = \sum_y \Delta_y \otimes \ket{y}\bra{y}.
\]
In \cite{DHWW12} a source code for it was \emph{defined} as a measurement 
(i.e., a qc-channel) taking values in a subset of $\mathcal{Y}^n$,
and as in Example~\ref{ex:classical-rd}, these codes can be understood
as quantum rate distortion codes in the above sense. We add as an aside
that in the limit of zero distortion, this model can be traced back to 
the work of Massar and Popescu \cite{MassarPopescu} and of Massar with 
one of us \cite{WinterMassar01}.

However, given a quantum rate distortion code in our present sense,
by the same logic as in Example~\ref{ex:classical-rd}, the
overall map $\mathcal{F}_n$ is without loss of generality a qc-channel, and the
coding part in Theorem~3 of \cite{DHWW12} achieves the same classical 
communication rates as Theorem~2 in \cite{DHW11} -- with the difference that
only shared randomness rather than entanglement is required, which
then can be removed entirely. With quantum communication assisted by
entanglement, we hence get half that rate, $R_{ea}^q(D) = \tfrac{1}{2}R^{qc}(D)$.

Conversely, assuming that the (unassisted) quantum rate distortion function
were strictly smaller than $R^{qc}(D)$, leads to a contradiction along the same lines
as in Example~\ref{ex:classical-rd}: we could use entanglement to replace
the compressed qubits by cbits at asymptotic exchange rate $1$, and by
superdense coding would obtain an entanglement-assisted rate distortion code
(of the same distortion) of rate $< \tfrac{1}{2}R^{qc}(D)$, contradicting
our conclusion in the previous paragraph.

In summary, the theory of quantum-to-classical rate distortion coding is
subsumed in the above general framework.
}
\end{example}

As reviewed above, in \cite{DHW11} we obtained
expressions for the quantum rate distortion function $R^q(D)$ and
the entanglement-assisted quantum rate distortion function  $R_{ea}^q(D)$ 
in terms of entropic quantities (note that going to general distortion observables
does not change the form of these results).
By unifying these results, we obtain a rate region characterizing
the quantum communication and entanglement consumption that is necessary and sufficient
for lossy compression of an IID quantum source.
This is given by Theorem~\ref{thm:trade-off-unassisted-assisted-qrd} below.

\begin{theorem}
\label{thm:trade-off-unassisted-assisted-qrd} For a memoryless quantum
information source defined by the density matrix 
$\rho \in \cD(\cH_A)$
with a
purification $|\psi^{\rho}_{RA}\rangle$, and any given distortion
$D\geq0$, the quantum rate distortion coding region for lossy source coding with
noiseless quantum communication, with the help of rate-limited shared
entanglement at rate $E$, is given by the union of the following regions,
letting $k$ become arbitrarily large:
\begin{align}
     Q & \geq \frac{1}{2k} I(R^k;B^k E_B)_{\omega}, \nonumber\\
 Q + E & \geq \frac{1}{k} H(B^k E_B)_{\omega}, \label{eq:trade-off-unassisted-EA-qrd}
\end{align}
where the entropic quantities are with respect to the following state:
\begin{equation}
\omega_{R^k B^k E_A E_B} \equiv V_{E^k \to E_A E_B} (U^{\mathcal{N}^{(k)}}_{A^k\rightarrow
B^k E^k}((\psi_{RA}^{\rho})^{\otimes k}) ), \label{eq:trade-off-code-state}
\end{equation}
and the union is over all isometric extensions $U^{\mathcal{N}^{(k)}}$ of CPTP maps
$\mathcal{N}^{(k)}$ such that $\overline{d}(\rho , \mathcal{N}^{(k)}) \leq D$
and isometries $V_{E^k \to E_A E_B}$.
\end{theorem}

\begin{proof}
Our proof of these bounds
requires just a slight modification of the proofs of the converse theorems
in Ref.~\cite{DHW11}. Indeed, consider the most general protocol
for rate-limited entanglement-assisted quantum rate distortion coding. The
protocol begins with the reference and Alice sharing the state $\left(
\psi_{RA}^{\rho}\right)^{\otimes n}$. Let $R^{n}$ denote the
reference's systems, and let $A^{n}$ denote Alice's systems. Alice and
Bob share entanglement in the systems $T_{A}$ and $T_{B}$ before communication
begins, and we suppose that the logarithm of the dimension of system $T_{B}$
is no larger than $nE$. Alice acts with an encoder (some general CPTP\ map)
on her systems $A^n$, obtaining a system $W$. She then sends $W$ to Bob, who 
subsequently feeds $W$ and $T_{B}$
into a decoder to produce the system $B^{n}$. 
By Stinespring's dilation theorem \cite{book2000mikeandike,W11}, we can simulate this protocol
by one in which Alice's encoder is replaced by an isometric extension of it,
with outputs $W$ and an environment $E_{1}$, and Bob's decoder is replaced by
an isometric extension of this decoder, with outputs $B^{n}$ and an
environment~$E_{2}$. At the end of the simulation, the state on systems
$R^n B^n E_1 E_2$ is a state of the form in (\ref{eq:trade-off-code-state}).

We can now obtain a lower bound on the rate $Q$ of 
quantum communication as follows:%
\begin{align*}
nQ  \equiv \log\left(\rm{dim }\cH_W\right) &  \geq H\left(  W\right) \\
&  =H\left(  WT_{B}\right)  -H\left(  T_{B}|W\right) \\
&  \geq H\left(  B^{n}E_{2}\right)  -nE .
\end{align*}
The first equality is
the entropy chain rule. The second inequality follows because $H\left(
WT_{B}\right)  =H\left(  B^{n}E_{2}\right)  $ (considering an isometric
extension of the decoder) and because conditioning cannot increase
entropy:$\ H\left(  T_{B}|W\right)  \leq H\left(  T_{B}\right)  \leq nE$.

The other bound results as follows:
\begin{align*}
2nQ & \geq 2H(W) \\
    & = I(W; R^n T_B E_1) \\
    & \geq I(W;R^n T_B) \\
    & = I(W T_B; R^n) + I(W;T_B) - I(R^n; T_B) \\
    & \geq I(W T_B; R^n) \\
    & = I(B^n E_2; R^n).
\end{align*}
The first equality follows from the fact that $ H(W) = H(R^n T_B E_1)$ and 
$ H(W R^n T_B E_1) = 0$ for a pure state on systems $W R^n T_B E_1$. The second inequality
results from applying the quantum data processing inequality. The second equality is an identity.
The third inequality follows because systems $R^n$ and $T_B$ are in a product state (implying $I(R^n; T_B) = 0$)
and from the fact that $I(W;T_B) \geq 0$. The final equality results
because entropy is invariant under isometries (in this case, the isometric extension of the decoder
taking systems $W T_B$ to $B^n E_2$. This proves the converse part of this theorem.

The achievability part of this theorem follows simply by picking a map that meets the
distortion constraint and applying Theorem~3b of \cite{BDHSW09}.
\end{proof}

It is worth remarking that in the case there is sufficient entanglement available, 
the above theorem reduces to the
entanglement-assisted quantum rate distortion function from Theorem~3 of \cite{DHW11},
while if there is no entanglement available, then the above theorem reduces to the 
entanglement of purification characterization from Theorem~5 of \cite{DHW11}.

\section{Classically-Assisted Quantum Rate Distortion}
\label{classical}

In this section, we consider quantum rate distortion coding in the presence of classical side information. As mentioned in the Introduction, this corresponds to the scenario in which Alice is allowed unlimited, forward classical communication to Bob to assist them in their compression-decompression task. 
We refer to the corresponding rate distortion function as the {\em{classically-assisted quantum rate distortion function}} and denote it by $R_{\rightarrow}^{q}(D)$ for a given distortion $D \ge 0$. It is defined analogously to $R^q(D)$ (see Section~\ref{quantum}), except that the encoding and decoding maps are 
now given by
$$\mathcal{E}_{n}:\mathcal{D}(\mathcal{H}_{A}^{\otimes n})\rightarrow
\mathcal{D}({\widetilde{\mathcal{H}}_{Q^{n}}}\otimes \cH_{X}),
$$
and
\[
\mathcal{D}_{n}:\mathcal{D}({\widetilde{\mathcal{H}}_{Q^{n}}} \otimes \cH_{X})\rightarrow
\mathcal{D}(\mathcal{H}_{A}^{\otimes n}),
\]
where $\cH_X$ denotes the Hilbert space associated with the classical information that Alice sends to Bob. 
Like the previous rate distortion functions, also $R_{\rightarrow}^q(D)$ is convex for
any given distortion observable $\Delta$.

We prove the following theorem, which gives an expression for $R_{\rightarrow}^{q}(D)$ in terms of the entanglement of formation defined in \reff{eof}. 

\begin{theorem}
\label{thm:cl-as-qrd} For a memoryless quantum information source defined by the
density matrix $\rho\in \cD({\cH_A})$, and any given distortion $D\geq0$, the quantum rate
distortion function assisted by unlimited, forward classical communication is given by%
\begin{equation}
R_{\rightarrow}^{q}\left(  D\right)  =\lim_{k\rightarrow\infty}\frac{1}{k}%
\min_{%
\genfrac{}{}{0pt}{}{\mathcal{N}^{(k)}\mathcal{\ }:}{{\overline{d}(\rho}{,\mathcal{N}^{(k)})\leq D}}%
}\ \left[  E_{F}(\rho^{\otimes k},\mathcal{N}^{(k)})\right]  ,
\label{eq:EoF-RD}%
\end{equation}
where $\mathcal{N}^{(k)}:\mathcal{D}(\mathcal{H}_{A}^{\otimes k}%
)\rightarrow\mathcal{D}(\mathcal{H}_{B}^{\otimes k})$ is a CPTP map, and
\begin{equation}
E_{F}(\rho,\mathcal{N})\equiv E_{F}(\omega_{RB}) \label{reg}%
\end{equation}
denotes the entanglement of formation of the state%
\begin{equation}
\omega_{RB}\equiv (\mathrm{{id}}_{R}\otimes\mathcal{N}_{A\rightarrow B}%
)(\psi_{RA}^{\rho}). \label{eq:channel-on-source-state}%
\end{equation}

\end{theorem}

\begin{proof}
The achievability part of the above theorem was essentially proven 
by Devetak and
Berger \cite{Devetak:2002it} for the particular case of a source of isotropic qubits, even though they did not express their result explicitly in the 
form of the entanglement of
formation. Moreover, they did not give a general converse proof. For the sake of 
completeness, we include a proof of achievability 
in addition to giving a proof of the converse.

To prove the achievability part, our approach is the same as that of Devetak and Berger \cite{Devetak:2002it}, namely, to exploit a variant of Schumacher
compression with classical communication \cite{Schumacher:1995dg}. To start with, consider $k=1$ on the RHS of \reff{eq:EoF-RD}, and fix the CPTP map
$\mathcal{N}\equiv \cN^{(1)}$ such that the minimization on the RHS of this equation is
achieved. Every Kraus decomposition of
this map $\mathcal{N}_{A\rightarrow B}$ as $\sum_{x}A_{x}\left(  \cdot\right)
A_{x}^{\dag}$, where $\sum_{x}A_{x}^{\dag}A_{x}=I$, leads to a pure-state
decomposition of the state $\omega_{RB}$:%
\[
\omega_{RB}=\sum_{x}\left(  I_{R}\otimes A_{x}\right)  (\psi_{RA}^{\rho
})(  I_{R}\otimes A_{x}^{\dag})  .
\]
In fact, all the pure-state decompositions of $\omega_{RB}$ and the Kraus decompositions
of $\mathcal{N}_{A\rightarrow B}$ are in one-to-one correspondence.
Note that each operator $(  I_{R}\otimes A_{x})  (\psi_{RA}^{\rho})(
I_{R}\otimes A_{x}^{\dag})  $ is of rank one, so that each normalized
version is a pure state. Consider the following extension of the above state:%
\begin{equation}
\omega_{RBX}=\sum_{x}(  I_{R}\otimes A_{x})  (\psi_{RA}^{\rho
})(  I_{R}\otimes A_{x}^{\dag})  \otimes\left\vert x\right\rangle
\left\langle x\right\vert _{X},\label{eq:RD-q-instrument}%
\end{equation}
where $X$ denotes a classical register. Note that the above state can be considered to result from the action of a quantum instrument on $\psi_{RA}^{\rho}$ since it has both a quantum and a classical part.
 
The entanglement of formation of the state $\omega_{RB}$ is then equal to%
\begin{equation}
E_{F}\left(  \omega_{RB}\right)  =\min_{\left\{  A_{x}\right\}  }H\left(
B|X\right)  _{\omega}, \label{eq:RD-EoF-achievability}%
\end{equation}
where the minimization is over the choice of Kraus operators. 

The protocol proceeds as follows. To start with, the reference and Alice share $n$ copies of $\psi_{RA}^\rho$, which is a purification of the source state
$\rho \in \cD(\cH_A)$. Alice determines the Kraus decomposition of the CPTP map $\mathcal{N}$
(chosen as described above) that minimizes the
conditional entropy $H\left(
B|X\right)  _{\omega}$. Henceforth, we denote the corresponding set of Kraus
operators simply as
$\left\{  A_{x}\right\}$. On every copy of the source state, she
performs the quantum instrument given by (\ref{eq:RD-q-instrument}). She
then measures the classical register $X$ of each copy of the resulting state
$\omega_{RBX}$, thus obtaining a classical sequence $x^n \equiv (x_1, x_2, \ldots, x_n)$, where $x_i$ denotes the outcome of measuring the $X$ register of the $i^{\text{th}}$ copy of $\omega_{RBX}^{\otimes n}$. From \reff{eq:RD-q-instrument} it is clear that the probability that Alice gets an outcome $x$ upon measuring an $X$ register is given by $p_{X}\left(  x\right)
\equiv\ $Tr$(A_{x}^{\dag}A_{x}\rho)$. In the large $n$
limit, with high probability, there are approximately $np_{X}\left(
x\right)  $ states in the length $n$ sequence such that the outcome of the
measurement is~$x$ (i.e., the sequence $x^{n}$ will be strongly typical
\cite{W11}\ with very high probability).
If the sequence Alice obtains is not strongly typical, she aborts the protocol. Otherwise, she groups together the states in the length $n$ sequence with the same
measurement outcome and performs Schumacher compression on each of these
blocks, compressing each block to approximately $np_{X}\left(  x\right)
H\left(  B\right)  _{\rho_{x}}$ qubits, where%
\[
\rho_{x} \equiv \frac{1}{p_{X}\left(  x\right)  }A_{x}\rho A_{x}^{\dag} \in \cD(\cH_B),
\]
and $H\left(  B\right)  _{\rho_{x}} \equiv H(\rho_x)$. She then sends these qubits to Bob, along with the classical sequence
$x^{n}$ representing her measurement outcomes.

Thus, the total rate at which she sends qubits to Bob is given by
(\ref{eq:RD-EoF-achievability}) because%
\[
\sum_{x}p_{X}\left(  x\right)  H\left(  B\right)  _{\rho_{x}}=H\left(
B|X\right)  _{\omega}.
\]
Conditional on the sequence
$x^{n}$ that he receives, Bob decompresses the qubits in each block (according to Schumacher
decompression) and finally discards the classical sequence. The result of this
protocol is that, for $n$ large enough, a state very close to 
$\omega_{RB}^{\otimes n}$ is shared between the reference and Bob. Of course,
in the above development, we analyzed the protocol by assuming that each block consists of
exactly $np_{X}\left( x\right)  $ states, but one can analyze this more carefully (see
Ref.~\cite{DS03}, for example).

One could then execute the above protocol by blocking $k$ of the states
together and by having the 
CPTP map to be of the form $\mathcal{N}%
^{(k)}: \cD(\cH_{A^{k}})\rightarrow  \cD(\cH_{B^{k}})$, (where $\cH_{A^{k}} = \cH_A^{\otimes k}$ and $\cH_{B^{k}} = \cH_B^{\otimes k}$) acting on each block of $k$ states. By letting
$k$ become large, such a protocol leads to the rate in (\ref{eq:EoF-RD}) for
classically-assisted quantum rate distortion coding.

The converse part of the theorem is proved as follows. The most general
protocol begins with many copies of the state $\psi_{RA}^\rho$ 
being shared between the reference and Alice. The most general map that Alice can perform is a quantum
instrument from $A^{n}$ to a quantum system $W$ and a classical system $M$.
Let this be described by a set of trace non-increasing maps $\left\{
\mathcal{E}_{m}\right\}$, with $\sum_m \cE_m= I$. She sends the quantum system $W$ and the classical
message $M$ to Bob. Hence the rate of quantum communication is given by $Q\equiv(1/n) \log\left({\rm{dim}} \cH_W\right)$. Bob then performs a CPTP map from $WM$ to $B^{n}$. For him, performing a CPTP map on a classical system $M$ and a
quantum system $W$ is equivalent to performing CPTP maps $\mathcal{D}_{m}$,
on the quantum system $W$, conditional on the value $m$ of the classical register $M$ (see, e.g.,
 \cite{Yard05a}\ or Section~4.4.8 of Ref.~\cite{W11}). Let $\sigma_{R^{n}B^{n}%
}$ denote the state shared by the reference and Bob at the end of the protocol, and let $\mathcal{M}%
_{A^{n}\rightarrow B^{n}}$ denote the classically-coordinated
encoding-decoding map:%
\begin{align*}
\sigma_{R^{n}B^{n}}  &  \equiv\left(  \text{id}_{R^{n}}\otimes\mathcal{M}%
_{A^{n}\rightarrow B^{n}}\right)  \left(  \left(  \psi_{RA}^{\rho}\right)
^{\otimes n}\right)  ,\\
\mathcal{M}_{A^{n}\rightarrow B^{n}}  &  \equiv\sum_{m}\mathcal{D}_{m}%
\circ\mathcal{E}_{m}.
\end{align*}


Note that the quantum instrument employed by Alice can be simulated by an isometry followed by a von Neumann measurement. Specifically, we can consider Alice to perform an isometry from $A^{n}$ to quantum
systems $W$, $M^\prime$, and an environment $E_{1}$, and then do a von
Neumann measurement on $M^\prime$ to get a classical system $M$. 
After tracing over
the environment $E_{1}$, the original instrument is recovered. 
However, without loss of generality, we could also consider Alice
to perform a von Neumann measurement of $E_{1}$, thus obtaining a 
classical system$~L$. Let $\omega$ denote the state at this point.
Observe that the joint state of $R^{n}$ and $W$ is pure, when
conditioned on the classical systems $L$ and $M$. Each decoding map
$\mathcal{D}_{m}$ can be simulated by Bob by performing an isometry $U_{m}$ from
$W$ to $B^{n}$ and an environment $E_{2}$. Bob could subsequently perform
a von Neumann measurement on $E_{2}$, thus obtaining a classical system $K$. Let $\sigma$ denote the state at the end of the protocol.
Note that the state on $R^{n}B^{n}$ is pure when conditioned on the
classical registers $MLK$. Figure~\ref{fig:classical-assist-QRD} depicts both the original general protocol and the simulation of it outlined in this paragraph.

\begin{figure*}[ptb]
\begin{center}
\includegraphics[
width=6.8in]{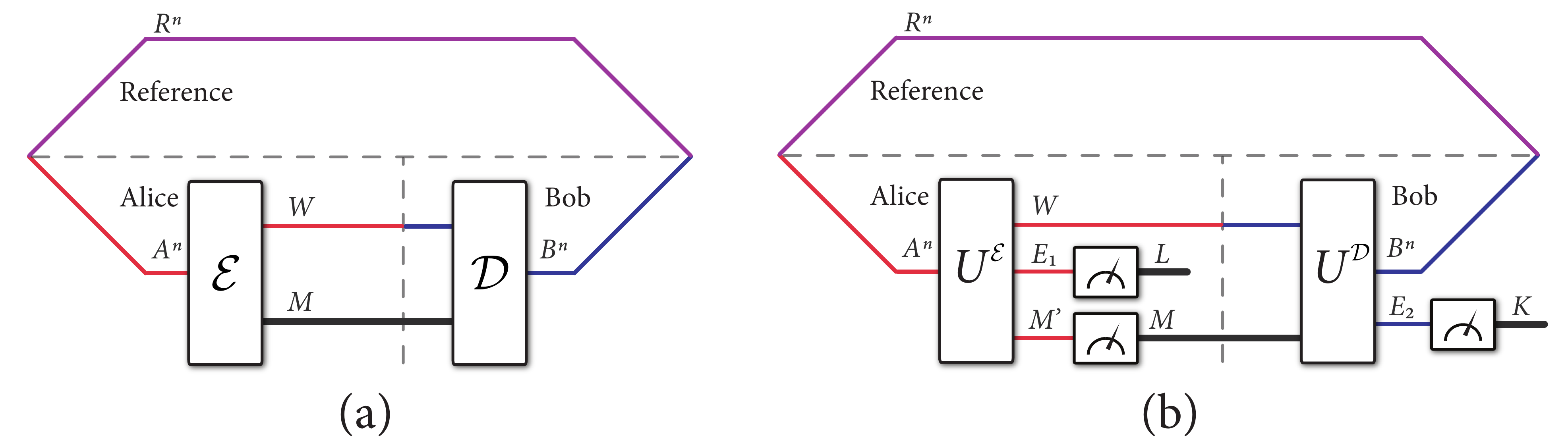}
\end{center}
\caption{\textbf{(a)} A general protocol for quantum rate distortion coding with the assistance of a forward classical side channel
from Alice to Bob. \textbf{(b)} A simulation of the general protocol in {{(a)}}, in which Alice and Bob respectively act with isometric 
extensions of the encoder and decoder in {{(a)}}.}
\label{fig:classical-assist-QRD}%
\end{figure*}

We can now prove a lower bound on the rate of 
classically-assisted lossy quantum data compression as follows:%
\begin{align*}
nQ \equiv \log \left({\rm{dim }} \cH_W \right)&  \geq H(W)_{\omega}\\
&  \geq H\left(  W|LM\right)_{\omega} \\
&  =H\left(  R^{n}|LM\right)_{\omega} \\
&  \geq H\left(  R^{n}|LMK\right)_{\sigma} \\
&  \geq E_{F}\left(  \sigma_{R^{n}B^{n}}\right) \\
&  \geq\min_{%
\genfrac{}{}{0pt}{}{\mathcal{N}^{(n)}\mathcal{\ }:}{{d(\rho}^{\otimes
n}{,\mathcal{N}^{(n)})\leq D}}%
}\ \left[  E_{F}(\rho^{\otimes n},\mathcal{N}^{(n)})\right].
\end{align*}
The second inequality follows because
conditioning on classical variables $L$ and $M$ (after Alice's simulation of
the encoding) cannot increase entropy. The first equality follows because (as
stated above) the joint state of $R^{n}$ and $W$ is pure, when conditioned on the
classical systems $L$ and~$M$. The third inequality follows again because
conditioning on the classical register $K$ cannot increase entropy (note that
this latter entropy is with respect to the state after Bob's simulation of the
decoder). Now, for the fourth inequality, as we stated above, the variables
$LMK$ induce a particular pure-state decomposition of the state on $R^{n}%
B^{n}$, and by the definition of entanglement of formation given in (\ref{eof}), the conditional
entropy of this particular pure-state decomposition cannot be larger than the
minimal one given by the entanglement of formation. The final bound follows
because the map $\sum_{m}\mathcal{D}_{m}\circ\mathcal{E}_{m}$ is a particular
CPTP map meeting the distortion constraint ${d(\rho}^{\otimes
n}{,\mathcal{N}^{(n)})\leq D}$, and thus the entanglement of formation
for the state resulting from this map cannot be larger than the entanglement
of formation of the state resulting from the optimal map meeting the
distortion constraint. Finally, we divide both sides of the above inequality by $n$ and take the limit as $n\to \infty$ to obtain the lower bound.
\end{proof}

We remark that the proof of the achievability part exploits a strategy
which is similar in spirit to that used in the proof of the reverse Shannon theorem (see Refs.~\cite{BDSW96,HHT01,BBCW11}, for example). 
In particular, we just pick the map that meets the distortion constraint and minimizes the entanglement of formation and simulate
this map using classical communication and Schumacher compression.  

%

\subsection{Bounds on the Quantum Rate Distortion Function for an Isotropic Qubit Source}

In this subsection we consider the original case of the distortion being based on the
entanglement fidelity, i.e., with distortion observable $\Delta = \1-\psi^\pi$, for
an isotropic qubit source, meaning that the source state
is a maximally mixed qubit state, $\pi \equiv \1/2$. 
The following theorem provides an exact expression for the entanglement-assisted quantum rate distortion function of an isotropic qubit source. 
\begin{theorem}
\label{thm:ea-qrd-isotropic}
For an isotropic qubit source, the entanglement-assisted quantum rate
distortion function is equal to
\[
  R_{ea}^q(D) = \begin{cases}
                  1-\tfrac{1}{2}H\left( \{1-D,\tfrac{D}{3},\tfrac{D}{3},\tfrac{D}{3}\}\right) 
                                                           & \text{ if } 0\leq D\leq \frac34, \\
                  0 & \text{ if } \frac34 \leq D \leq 1,
                \end{cases}
\]
where we have used the notation $H\left( \{\cdot\}\right)$ to denote
the Shannon entropy of the probability distribution inside the braces~$\{\cdot\}$.
\end{theorem}

\begin{proof}
First, recall the entanglement-assisted rate distortion function from Theorem~3 of \cite{DHW11}:%
\begin{equation}
R_{ea}^{q}(D)
  = \frac{1}{2}\min_{d\left(\rho,\mathcal{N}\right) \leq D}\ I\left(R;B\right) _{\omega},
\label{eq:ea-rate-d-func}%
\end{equation}
where the distortion measure $d\left(  \rho,\mathcal{N}\right)  $\ is related
to the entanglement fidelity:
\[
d\left(  \rho,\mathcal{N}\right) \equiv 1-F_{e}\left(  \rho,\mathcal{N} \right)  .
\]
The mutual information is with respect to the state
\[
\omega_{RB}\equiv\left(  \text{id}_{R}\otimes\mathcal{N}_{A^{\prime
}\rightarrow B}\right)  \left(  \psi_{RA^{\prime}}^{\rho}\right)  ,
\]
where $\psi_{RA^{\prime}}^{\rho}$ is a purification of the source state $\rho
$. For an isotropic qubit source, we have $\rho
=\pi\equiv \1 /2$ and $\psi_{RA^{\prime}}^{\rho} = \Phi_{RA^{\prime}}$ (a maximally entangled state). For any channel $\mathcal{N}$, it has a Kraus decomposition
as follows:%
\[
\mathcal{N}\left(  \rho\right)  =\sum_{x}A_{x}\rho A_{x}^{\dag}.
\]
We also have the well-known formula for the entanglement fidelity (see, e.g., \cite{W11}):%
\[
F_{e}\left(  \rho,\mathcal{N}\right)  =\sum_{x}\left\vert \text{Tr}(
\rho A_{x})  \right\vert ^{2}.
\]

Now suppose that there is some channel $\mathcal{N}$ achieving the minimum in
(\ref{eq:ea-rate-d-func}), with Kraus
operators $\left\{  A_{x}\right\}$. 
Consider the channel $\mathcal{N}_{i}$ defined as follows:
\[
\mathcal{N}_{i}\left(  \rho\right)  \equiv\sigma_{i}^{\dag}\mathcal{N}\left(
\sigma_{i}\rho\sigma_{i}^{\dag}\right) \sigma_{i},
\]
where $\sigma_{i}$, $i=0,1,\ldots,11$ are the Clifford unitaries on a single
qubit (given explicitly, e.g., in Appendix~A of \cite{BDSW96}). Thus, its Kraus
operators are $\left\{  \sigma_{i}^{\dag}A_{x}\sigma_{i}\right\}$ for any
fixed $i$. For an isotropic qubit source, the channel $\mathcal{N}_{i}$ has
the same entanglement fidelity as the original channel because
\begin{multline}
F_{e}\left(  \pi,\mathcal{N}_{i}\right)  =\sum_{x}\left\vert \text{Tr}\left\{
\pi\sigma_{i}^{\dag}A_{x}\sigma_{i}\right\}  \right\vert ^{2}  =\frac{1}{4}%
\sum_{x}\left\vert \text{Tr}\left\{  \sigma_{i}^{\dag}A_{x}\sigma_{i}\right\}
\right\vert ^{2} \\ =\frac{1}{4}\sum_{x}\left\vert \text{Tr}\left\{
A_{x}\right\}  \right\vert ^{2}=F_{e}\left(  \pi,\mathcal{N}\right).
\label{eq:ent-fid-calc}
\end{multline}
Let $\mathcal{N}_{\text{tw}}$ denote the \textquotedblleft twirled
version\textquotedblright\ of $\mathcal{N}$:
\[
\mathcal{N}_{\text{tw}}\left(  \rho\right)  \equiv\frac{1}{12}\sum
_{i}\mathcal{N}_{i}\left(\rho\right).
\]
A similar calculation as in (\ref{eq:ent-fid-calc}) reveals that the
\textquotedblleft twirled version\textquotedblright\ of the channel
$\mathcal{N}$ has an entanglement fidelity equal to $F_{e}\left(
\pi,\mathcal{N}\right)$. Also, it is well known that the Clifford twirled channel is
equal to a depolarizing channel (a probabilistic mixture of the identity
channel and the constant channel mapping every input state to the
maximally mixed state $\pi = \tfrac{1}{2}\1$, see, e.g., \cite{BDSW96,VollbrechtWerner01,ABE10}). 
Now, each of the channels $\mathcal{N}_{i}$ leads to the same
mutual information as the original channel $\mathcal{N}$, in the sense that%
\[
I\left(R;B\right)_{\omega} = I\left(R;B\right)_{\omega_{i}},
\]
where%
\begin{align*}
\omega & \equiv \omega_{RB}\equiv\left(  \text{id}_{R}\otimes\mathcal{N}_{A^{\prime
}\rightarrow B}\right)  \left(  \Phi_{RA^{\prime}}\right)  ,\\
\omega_{i} & \equiv \left(  \omega_{i}\right)  _{RB}  \equiv \left(  \text{id}_{R}\otimes\left(
\mathcal{N}_{i}\right)  _{A^{\prime}\rightarrow B}\right)  \left(
\Phi_{RA^{\prime}}\right)  .
\end{align*}
This is due to the fact that, for a maximally entangled state 
$\left\vert \Phi_{RA^{\prime}}\right\rangle$,
$\1_{R}\otimes\left(  \sigma_{i}\right)
_{A^{\prime}}\left\vert \Phi_{RA^{\prime}}\right\rangle=\left(  \sigma
_{i}^{T}\right)  _{R}\otimes \1_{A^{\prime}}\left\vert \Phi_{RA^{\prime}}\right\rangle$ 
(where $\sigma_i^T$ denotes the transpose of $\sigma_i$), and because the von Neumann entropy is
invariant under unitaries. However, we know that the twirled channel
cannot have a mutual information larger than the original channel's, due to
the convexity of mutual information with respect to the states 
$\left(\omega_{i}\right)_{RB}$ :%
\[
I\left(  R;B\right)_{\omega}
  =\frac{1}{12}\sum_{i}I\left(R;B\right)_{\omega_{i}}\geq I\left(  R;B\right)_{\omega_{\text{tw}}},
\]
where
\[
\omega_{\text{tw}}\equiv\left(  \text{id}_{R}\otimes\left(  \mathcal{N}_{\text{tw}}\right)_{A^{\prime}\rightarrow B}\right)  \left(  \Phi_{RA^{\prime}}\right)  .
\]
This proves that the channel optimizing the expression in
(\ref{eq:ea-rate-d-func}) for an isotropic qubit source is a depolarizing channel
$\cN_p$, hence of the form
\[
\mathcal{N}_p\left(\rho\right)  = (1-p)\rho + \frac{p}{3}
\left(
\sigma_X \,\rho\, \sigma_X + \sigma_Y \,\rho\, \sigma_Y + \sigma_Z \,\rho\, \sigma_Z \right).
\]
For these channels, a simple calculation reveals that their entanglement
fidelity for an isotropic qubit source is equal to $p$, because the
non-identity Pauli operators are traceless. Thus, for a given distortion $D$,
the channel achieving the minimum mutual information is a depolarizing channel with
$p \leq D$. The latter is given by
\[
  1-\tfrac{1}{2}H\left(\{ 1-p, \tfrac{p}{3}, \tfrac{p}{3}, \tfrac{p}{3} \}\right),
\]
thus finishing the proof.
\end{proof}

In prior work  \cite{Devetak:2002it}, Devetak and Berger showed that the classically-assisted
quantum rate distortion function in Theorem~\ref{thm:cl-as-qrd} significantly simplifies
for an isotropic qubit source. For convenience,
we restate their result as the following theorem and provide a simple proof of it below.

\begin{theorem}[Devetak and Berger \cite{Devetak:2002it}]
\label{thm:cl-a-qrd-isotropic}
The classically-assisted
quantum rate distortion function for an isotropic qubit source is equal to the following expression:
\[
R_{\rightarrow}^{q}\left(D\right)  =\left\{
\begin{array}
[c]{cc}%
h_{2}\left(  \tfrac{1}{2}+\sqrt{D\left(  1-D\right)  }\right)   & : \, 0 \leq D<\tfrac{1}%
{2}\\
0 & : \, \tfrac{1}{2}\leq D\leq1
\end{array}
\right.
\]
In the above, $h_2(p) \equiv - p \log p - (1-p) \log (1-p)$ is the binary entropy
for any probability $p$. 
\end{theorem}

\begin{proof}
First, recall the general expression for the classically-assisted quantum rate
distortion function from Theorem~\ref{thm:cl-as-qrd}. 
Devetak and Berger have shown that this
expression assumes a single-letter form for the case of an 
isotropic qubit source \cite{Devetak:2002it}, and we
do not reproduce the proof of this statement here.

So, the expression for the
classically-assisted quantum rate distortion function in this case reduces to
\begin{equation}
\min_{d\left(  \pi,\mathcal{N}\right)  \leq D}\min_{\left\{  A_{x}\right\}
}H\left(  B|X\right)_{\omega},\label{eq:EoF-simple}%
\end{equation}
where the conditional entropy is with respect to the following state:%
\[
\omega_{RBX}\equiv\sum_{x}\left(  I_{R}\otimes A_{x}\right)  \left(
\Phi_{RA^{\prime}}\right)  (  I_{R}\otimes A_{x}^{\dag})
\otimes\left\vert x\right\rangle \left\langle x\right\vert _{X},
\]
and the operators $\left\{  A_{x}\right\}_{x}$ are the Kraus operators for a
channel $\mathcal{N}$ meeting the distortion constraint. For simplicity, let
us denote the optimal channel meeting the distortion constraint in
(\ref{eq:EoF-simple}) as $\mathcal{N}$ and the optimal Kraus decomposition for
the entanglement of formation as $\left\{  A_{x}\right\}_{x}$, so that
$d\left(  \pi,\mathcal{N}\right) \leq D$. By the same argument as in the previous
theorem, the channel with the set of Kraus operators $\left\{  \sigma_{i}^{\dag}%
A_{x}\sigma_{i}\right\}_{x}$ for a fixed $i$ causes the same distortion
$D$ to an isotropic qubit source while having the same value for the
entanglement of formation. Also, by the same argument, the twirled channel
with Kraus operators $\left\{  \sigma_{i}^{\dag}A_{x}\sigma_{i}/\sqrt
{12}\right\}_{x,i}$ causes the same distortion $D$\ as the optimal channel,
but this channel can have only a lower value of the entanglement of formation, due to
the convexity of the entanglement of formation \cite{BDSW96}. 
Now, the twirled channel is a depolarizing channel causing
distortion $\leq D$ to the source, implying that its effect on a maximally
entangled state is to produce an isotropic state, i.e., a mixture of Bell 
states of the following form:
\[
 (1-p) \Phi_{RA^{\prime}} + \tfrac{p}{3} \Psi_{RA^{\prime}}^{+}
  + \tfrac{p}{3} \Psi_{RA^{\prime}}^{-} + \tfrac{p}{3} \Phi_{RA^{\prime}}^{-},
\]
where $p \leq D$. In the above mixture, it must be the case that
$1-D$ is larger than all of the other components whenever $D<1/2$. 
In this case, it is well known that the entanglement of
formation of such a Bell mixture is equal to the following expression \cite{BDSW96}:
\[
h_{2}\left(  \tfrac{1}{2}+\sqrt{D\left(  1-D\right)  }\right)  .
\]

For $D\geq1/2$, the strategy requiring no quantum communication is very
simple, implying that there is no need to explicitly evaluate the expression
in the theorem statement. For every copy of the source, Alice just measures it in the basis
$\left\{  \left\vert 0\right\rangle ,\left\vert 1\right\rangle \right\}  $ and
sends the measurement outcome to Bob over the classical channel. Bob then
prepares the state $\left\vert 0\right\rangle $ or $\left\vert 1\right\rangle
$ depending on what he receives from Alice, and he forgets what Alice sent to
him. This procedure prepares the dephased state $\frac{1}{2}\left(  \left\vert
0\right\rangle \left\langle 0\right\vert _{R}\otimes\left\vert 0\right\rangle
\left\langle 0\right\vert _{B}+\left\vert 1\right\rangle \left\langle
1\right\vert _{R}\otimes\left\vert 1\right\rangle \left\langle 1\right\vert
_{B}\right)  $ shared between Bob and the reference, which has distortion
$1/2$ from the maximally entangled state. To achieve an even larger distortion
 with no quantum communication (if one so wishes), Bob could just depolarize
his state locally.
\end{proof}
\begin{figure}[ptb]
\begin{center}
\includegraphics[
width=3.5in
]{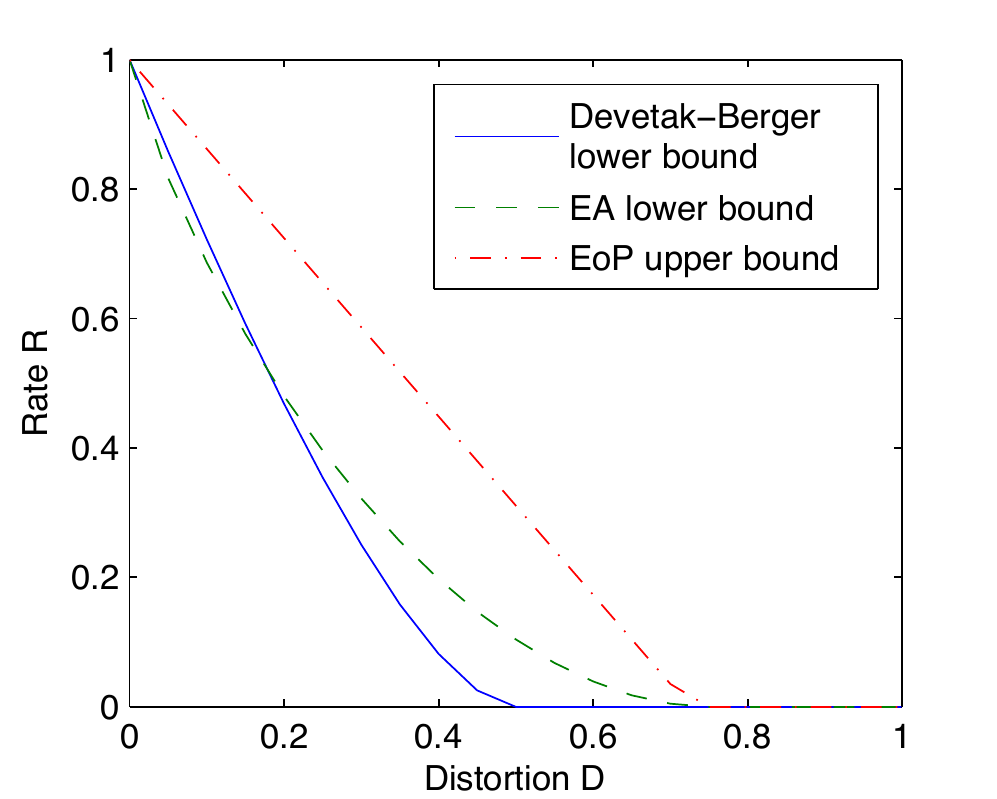}
\end{center}
\caption{Bounds on the unassisted quantum rate distortion function for an isotropic qubit source.
We have lower bounds from the classically-assisted and entanglement-assisted (EA) quantum rate distortion function of this source.
The convexified entanglement of purification (EoP) provides an upper bound on the quantum rate distortion function
of this source.}
\label{fig:RD-compare}
\end{figure}

The expressions from
Theorems~\ref{thm:ea-qrd-isotropic} and \ref{thm:cl-a-qrd-isotropic} give two lower bounds
on the unassisted quantum rate distortion function of an isotropic qubit source, plotted in Figure~\ref{fig:RD-compare}.

We can also obtain an upper bound on the quantum rate distortion function for the isotropic qubit source
by computing the unregularized entanglement of purification bound from Theorem~5 of Ref.~\cite{DHW11}.
In particular, a strategy to achieve the unregularized bound is as follows.
The protocol begins with Alice and the reference sharing
many copies of a maximally entangled state $\Phi_{RA'}$ (the reduction of each of these to Alice's systems
is a maximally mixed state). Given a distortion constraint $D$, Alice applies some
isometric extension of the following depolarizing channel to each of her systems:
$$
\mathcal{N}_D(\rho) \equiv (1-D) \rho + \tfrac{D}{3} \left(
\sigma_X \,\rho\, \sigma_X + \sigma_Y \,\rho\, \sigma_Y + \sigma_Z \,\rho\, \sigma_Z \right),
$$
resulting in the following state shared between the reference and Alice:
\begin{equation}
\left(  1-D\right)  \Phi_{RA^{\prime}}+ \tfrac{D}{3} \left( \Psi_{RA^{\prime}}^{+}%
+ \Psi_{RA^{\prime}}^{-}+ \Phi_{RA^{\prime}}^{-} \right). \label{eq:Werner}
\end{equation}
Alice then Schumacher compresses the output of the depolarizing channel and some share of the environment
(which she possesses since she performs the isometric extension), and she can do this at a rate equal
to the entanglement of purification of the above state. Furthermore, Alice and Bob can time-share
between any two strategies of this form, implying that the rates and distortions of the time-shared protocol
combine as in Remark~\ref{rem:convexity}. The authors of Ref.~\cite{THLD02}
have already numerically calculated
the entanglement of purification of the state in (\ref{eq:Werner})
in Figure~1 of their paper. As such, the convex hull of
their plot serves
equally well as an upper bound on the quantum rate distortion function of the isotropic qubit
source, and we have reproduced this plot in our Figure~\ref{fig:RD-compare}
above.\footnote{The only change we need to
make to their plot to suit our purposes is to flip it with respect to the horizontal axis, eliminate
values of the entanglement of purification beyond $D=3/4$, and take the convex hull of
the resulting curve. The maximally mixed state on Alice and the reference
is a state that meets the distortion constraint at $D=3/4$ so that at this distortion or beyond, there is
no quantum communication needed---Alice simply discards her qubits and Bob prepares the maximally mixed state at his end.}

\section{Quantum Rate Distortion in the Presence of Quantum Side Information}
\label{sec-qrd-qsi}

In this section we study quantum rate distortion in the case in which 
Bob has some  quantum side information (QSI) about the source, as an auxiliary resource. As mentioned in the Introduction, this corresponds to the following setting:
Suppose a third party (say, Charlie) maps the source state~$\rho$ via some isometry to a bipartite state $\rho_{AB}$ and distributes the systems $A$ and $B$ to Alice and Bob, respectively. The goal is for Alice to transfer her system $A$ to Bob, up to some given distortion, using as few qubits as possible. The system $B$, which is in Bob's possession, acts as the quantum side information, and he can make use of it in his decompression task. It is required that the protocol causes only a negligible disturbance to the state of the reference system and Bob's quantum side information, in case Bob might want to use his system in some subsequent quantum information processing task. The above problem is a quantum generalization of the Wyner-Ziv \cite{WZ76} problem and is also a natural extension of the work
of Luo and Devetak \cite{LD09} which dealt with classical rate distortion theory in the presence of QSI (and thus considered Alice to receive a classical system instead of a quantum one).

The rate distortion function, which we denote as $R^q_{qsi}(D)$ for any given distortion $D\ge 0$, is then the minimum rate of quantum communication required for this task, evaluated in the limit in which Alice and Bob share asymptotically many copies of the state $\rho_{AB}$. It is defined analogously to $R^q(D)$ (see Section~\ref{quantum}), except that the encoding and decoding maps are 
now given by
$$\mathcal{E}_{n}:\mathcal{D}(\mathcal{H}_{A}^{\otimes n})\rightarrow
\mathcal{D}({\widetilde{\mathcal{H}}_{Q^{n}}}),
$$
and
\[
\mathcal{D}_{n}:\mathcal{D}({\widetilde{\mathcal{H}}_{Q^{n}}} \otimes \cH_{B}^{\otimes n})\rightarrow
\mathcal{D}(\mathcal{H}_{A}^{\otimes n}),
\]
where $\cH_{B}^{\otimes n}$ denotes the Hilbert space associated with Bob's QSI.

Theorem~\ref{thm:QSI-QRD-no-disturb} of Section~\ref{qsi} gives an expression for $R^q_{qsi}(D)$. Before going over to it, we briefly recall an important protocol of quantum information theory, namely, quantum state redistribution \cite{DY08,YD09}. After doing so,
 we then employ it in Section~\ref{qrst} to develop a quantum reverse Shannon theorem in the presence of QSI---the main tool that we use to prove Theorem~\ref{thm:QSI-QRD-no-disturb}.
\medskip

\subsection{Quantum State Redistribution} 

Quantum state redistribution is an important protocol in quantum information theory \cite{DY08,YD09} and is defined as follows. Alice and Bob share many copies of a tripartite state $\rho_{ABC}$,
 where Alice holds the systems labeled by $A$ and $C$, and Bob holds the systems labeled by $B$.
 Let the state $\rho_{ABC}$ be purified by a reference system $R$, the pure state being denoted as $\psi_{ABCR}$. The task is for Alice to
 transfer the systems labeled by $C$ to Bob,
 while keeping the overall purification $\psi_{ABCR}$ approximately unchanged (possibly with the help of prior shared entanglement). The quantity of interest is the minimum rate of quantum communication from Alice to Bob needed to accomplish this task. The rate is evaluated in the limit of asymptotically many copies of the state $\rho_{ABC}$ that is initially shared between Alice and Bob. 

Let $Q$ and $E$ denote the rates of quantum communication and entanglement consumption,\footnote{A negative entanglement consumption rate implies that entanglement is instead generated! The reader should keep this in mind any time we refer to the ``entanglement consumption rate'' of a protocol.} respectively, required to achieve quantum state redistribution. Devetak and Yard~\cite{DY08} proved that the corresponding resource inequality is given as follows:
\be\label{resource}
\psi_{AC|B|R} + Q[q\to q] + E[qq] \ge \psi_{A|CB|R},
\ee
if and only if $Q$ and $E$ satisfy the following inequalities:
\begin{align}
Q & \ge \tfrac{1}{2} I(R;C|B)_\psi \nonumber \\
Q + E & \ge H(C|B)_\psi.\label{ineq}
\end{align}
The meaning of the resource inequality in \eqref{resource} is that, for $n$ large enough, $Q$
qubits of quantum communication and
$E_1$ ebits of entanglement are
sufficient to transfer all $n$ of the $C$ systems from Alice to Bob, while
generating an additional $E_2$
ebits of entanglement, such that $E = E_1 - E_2$. Moreover, the fidelity of this protocol is equal to one in the
asymptotic limit, and 
$\psi_{AC|B|R}$ and $\psi_{A|CB|R}$ denote the states before and after the protocol because Alice begins by holding the systems labeled by $AC$ and ends by holding only the systems labeled by $A$.
From \reff{resource}-\reff{ineq}, we infer that if 
$\frac{1}{2} I(R;C|B)_\psi >  H(C|B)_\psi$ then the protocol redistributes the system $C$ to Bob as well as generates entanglement. In this case, the resource inequality takes the form:
\begin{multline}
\psi_{AC|B|R} +  \tfrac{1}{2} I(R;C|B)_\psi [q\to q] \\ \ge 
\psi_{A|CB|R} + \left( \tfrac{1}{2} I(R;C|B)_\psi  -  H(C|B)_\psi\right)[qq],
\end{multline}
which can be equivalently written as
\begin{multline}
\label{resource-ent} 
\psi_{AC|B|R} +  \tfrac{1}{2} I(R;C|B)_\psi [q\to q] \\ \ge \psi_{A|CB|R} +  \tfrac{1}{2} \left( I(B;C)_\psi  -  I(A;C)_\psi\right)[qq],
\end{multline}
which follows from the definitions of the conditional mutual information and the conditional entropy, duality of conditional entropy (Lemma~\ref{dual_cond_ent}), and the fact that $\psi_{ABCR}$ is a pure state.

If, in contrast, $\tfrac{1}{2} I(R;C|B)_\psi <  H(C|B)_\psi$, then entanglement is consumed in order to achieve state redistribution, and the resource inequality can be written as 
\begin{multline}
\label{resource-ent2}
\psi_{AC|B|R} +  \tfrac{1}{2} I(R;C|B)_\psi [q\to q] \\
+ \tfrac{1}{2} \left(I(A;C)_\psi-  I(B;C)_\psi  \right)[qq] \ge \psi_{A|CB|R}\,\,.
\end{multline}

From \reff{resource-ent} and \reff{resource-ent2} it follows that if Alice and Bob have no prior shared entanglement, then they could still achieve their
task of state redistribution: if $\frac{1}{2} I(R;C|B)_\psi \geq  H(C|B)_\psi$ then they also generate entanglement, whereas if $\frac{1}{2} I(R;C|B)_\psi \le  H(C|B)_\psi$ then they need to invest quantum communication at a rate of $\approx H(C|B)_\psi$ qubits per copy of the source in order to generate the entanglement that the protocol corresponding to \reff{resource-ent2} requires.

\subsection{Quantum Reverse Shannon Theorem with Quantum Side Information}
\label{qrst}

As mentioned above, before moving on to quantum rate distortion theorems with QSI, we prove two
quantum reverse Shannon theorems with quantum side information. 

Quantum reverse Shannon theorems \cite{BDHSW09,BCR09} deal with the simulation of noisy quantum channels between two parties (Alice and Bob, say), with the aid of noiseless resources, such as
 prior shared entanglement and quantum communication. Here we consider the situation in which Bob has quantum side information as an auxiliary resource, which he can employ in this simulation task.

In particular, we consider Alice and Bob to share many (say, $n$) copies of a bipartite state $\rho_{AB}$, the systems $A$ being with Alice and $B$ being with Bob, the latter acting as the quantum side information. In addition, Alice and Bob can share entanglement with each other. The aim is for Alice and Bob to simulate many instances of a noisy channel $\cN_{A \rightarrow B'}$, such that Bob receives the output systems. The quantities of interest are the minimum rates of quantum communication and entanglement consumption required for this purpose.

We consider two different scenarios. In the first, which is referred to as a {\em{feedback simulation}}, the environments of the simulated channels are required to be in Alice's possession. In contrast, in the second scenario, which is referred to as a {\em{non-feedback simulation}}, no such requirement is imposed. The minimum rates of quantum communication and entanglement consumption that are required 
in these two scenarios are given by Theorems~\ref{thm:QRST-QSI-feedback} and 
\ref{thm:QRST-QSI-non-feedback}, respectively.


These
theorems are generalizations of Theorems~3a, 3b, and 3c of Ref.~\cite{BDHSW09}
and are interesting in their own right. Furthermore, both theorems are useful in establishing quantum rate distortion theorems that
exploit quantum side information.

\begin{theorem}
[Feedback QRST with QSI]%
\label{thm:QRST-QSI-feedback}
If Alice and Bob share many copies of a state $\rho_{AB}$,
then they can achieve a feedback simulation of many instances of a noisy channel $\cN_{A\rightarrow B'}$
if and only if the rates of quantum communication and entanglement consumption 
are in the following rate region:%
\begin{align}
Q+E  &  \geq H\left(  B^{\prime}|B\right)  _{\omega}%
,\label{eq:feedback-QRST-QSI-1}\\
Q  &  \geq\tfrac{1}{2}I\left(  R;B^{\prime}|B\right)  _{\omega},
\label{eq:feedback-QRST-QSI-2}%
\end{align}
where%
\[
\omega_{RB^{\prime}B}\equiv\mathcal{N}_{A\rightarrow B^{\prime}}\left(
\phi_{RAB}^{\rho}\right)  ,
\]
$\phi_{RAB}^{\rho}$ is a purification of $\rho_{AB}$. Equivalently, we can
write the rate of quantum communication as a function of the entanglement
consumption rate $E$ as follows:%
\[
Q_{f,qsi}\left(  E\right)  =\max\left\{  \tfrac{1}{2}I\left(  R;B^{\prime
}|B\right)  _{\omega},\ H\left(  B^{\prime}|B\right)  _{\omega}-E\right\}  .
\]
(Recall that $E$ can be either positive or negative depending on whether the protocol
consumes or generates entanglement, respectively.)
The subscript $f$ denotes that the rate corresponds to a feedback simulation.
In particular, if there is no entanglement available ($E=0$), then the optimal
rate of quantum communication is equal to%
\[
Q_{f, qsi}(0) =\max\left\{  \tfrac{1}{2}I\left(  R;B^{\prime}|B\right)  _{\omega
},\ H\left(  B^{\prime}|B\right)  _{\omega}\right\}  .
\]

\end{theorem}

\begin{theorem}
[Non-Feedback QRST with QSI]%
\label{thm:QRST-QSI-non-feedback} If Alice and Bob share  many copies of a state $\rho_{AB}$,
then the minimum rates of quantum communication and entanglement consumption 
that they need for a non-feedback simulation of many instances of a noisy channel $\cN_{A\rightarrow B'}$ are given by the union of the following rate regions:%
\begin{align}
Q+E &  \geq\frac{1}{k}H\left(  B^{\prime \, k}E_{B}|B^{k}\right)  _{\omega
},\label{eq:non-feedback-QRST-QSI-1}\\
Q &  \geq\frac{1}{2k}I\left(  R^{k};B^{\prime \, k}E_{B}|B^{k}\right)  _{\omega
},\label{eq:non-feedback-QRST-QSI-2}%
\end{align}
where $k$ is an arbitrary positive integer and the union is with respect to
all states $\omega$ of the following form:%
\begin{equation}
\omega_{R^{k} E_A E_B B^{\prime \, k} B^k}\equiv V_{E^{k}\rightarrow E_{A}E_{B}}\left(  \left(
U^{\mathcal{N}}_{A\rightarrow B^{\prime}E}\left(  \phi_{RAB}^{\rho}\right)
\right)  ^{\otimes k}\right)  ,\label{eq:code-state-QRST-QSI}%
\end{equation}
$\phi_{RAB}^{\rho}$ is a purification of $\rho_{AB}$, $U^{\mathcal{N}%
}_{A\rightarrow B^{\prime}E}$ is some isometric extension of the channel
$\mathcal{N}_{A\rightarrow B^{\prime}}$, and $V_{E^{k}\rightarrow E_{A}E_{B}}$
is an arbitrary isometry that splits the $k$ environment systems $E^{k}$
into two parts $E_{A}$ and $E_{B}$. Equivalently, we can write the rate of
quantum communication as a function of the entanglement consumption rate $E$
as follows:%
\begin{multline}
Q_{qsi}\left(  E\right)  =\liminf_{k\rightarrow\infty,\exists V}
\max\bigg\{  \frac{1}{2k}I\left(  R^{k};B^{\prime \,
k}E_{B}|B^{k}\right)  _{\omega},
\\ \ \frac{1}{k}H\left(  B^{\prime \, k}E_{B}%
|B^{k}\right)  _{\omega}-E\bigg\}  .
\end{multline}
In particular, if there is no entanglement available ($E=0$), then the optimal
rate of quantum communication is equal to%
\begin{multline}
Q_{qsi}(0)=\liminf_{k\rightarrow\infty,\exists V}
\max\bigg\{  \frac{1}{2k}I\left(  R^{k};B^{\prime \, k}E_{B}|B^{k}\right)
_{\omega},\\ \ \frac{1}{k}H\left(  B^{\prime \, k}E_{B}|B^{k}\right)  _{\omega
}\bigg\}  .
\end{multline}

\end{theorem}

\medskip

\begin{proof}
[Proof of Theorems~\ref{thm:QRST-QSI-feedback} and
\ref{thm:QRST-QSI-non-feedback}]We prove the two theorems using similar
arguments. The achievability parts of both of the above theorems follow
directly by applying the protocol of Devetak and Yard for quantum
state redistribution \cite{DY08,YD09,PhysRevA.78.030302}. We prove
the achievability part of the non-feedback theorem first and then 
argue how the feedback version {{is a special case of it}}. 

To start with, Alice, Bob, and the reference share $n$ copies of the state $\phi_{RAB}^{\rho}$, which is a
purification of the state $\rho_{AB}$. Alice locally applies an isometric
extension $U^{\mathcal{N}}_{A\rightarrow B^{\prime}E}$ of the channel
$\mathcal{N}_{A\rightarrow B^{\prime}}$ to each system $A$ in her possession,
and then applies an environment-splitting
isometry $V_{E\rightarrow E_{A}E_{B}}$ to each system $E$ that results.
At this
point, the three parties share $n$ copies of the following pure state:%
\[
\omega_{RE_{A}E_{B}B^{\prime}B}\equiv V_{E\rightarrow E_{A}E_{B}}\left(
U^{\mathcal{N}}_{A\rightarrow B^{\prime}E}\left(  \phi_{RAB}^{\rho}\right)
\right)  ,
\]
where the reference has $R$, Alice has $E_{A}E_{B}B^{\prime}$, and Bob has
$B$. Alice would like to transmit all of her $E_{B}B^{\prime}$ systems to Bob,
and she can do this by using the protocol of quantum state redistribution. 

By setting $C = B'E_B$, $A=E_A$, and $\psi_{RACB} = \omega_{RE_{A}E_{B}B^{\prime}B}$ in \reff{resource}-\reff{resource-ent2}, we infer that the following rate region is achievable with
quantum state redistribution:%
%
%
%
\begin{align*}
Q  &  \geq\tfrac{1}{2}I\left(  R;B^{\prime}E_{B}|B\right)  _{\omega},\\
Q+E  &  \geq H\left(  B^{\prime}E_{B}|B\right)  _{\omega}.
\end{align*}
%
%
Now, a protocol that achieves the regularized rate region in
\reff{eq:non-feedback-QRST-QSI-1}-\reff{eq:non-feedback-QRST-QSI-2}, is 
very similar, but Alice and Bob instead act on blocks of $k$
states at a time. That is, they share $n$ copies of the state $\left(
\phi_{RAB}^{\rho}\right)  ^{\otimes k}$ (if they are allowed access to an
arbitrary number of shares of this state, then they can block them in this
way). Alice applies $n$ instances of the isometry $\left(  U^{\mathcal{N}%
}_{A\rightarrow B^{\prime}E}\right)  ^{\otimes k}$ to her systems $A^n$
and then applies an environment
splitting isometry $V_{E^{k}\rightarrow E_{A}E_{B}}$ to each $E^{k}$ resulting
from the previous step. By the same arguments as given above, the rate region
in \reff{eq:non-feedback-QRST-QSI-1}-\reff{eq:non-feedback-QRST-QSI-2} is
achievable, where the division by $k$ is needed to obtain the rates.


The achievability of the feedback protocol in
Theorem~\ref{thm:QRST-QSI-feedback} follows as a special case of the above. In
particular, there is no splitting of the environment into two parts, so that
Alice merely redistributes the $B^{\prime}$ system to Bob. Thus, the rate region in
\reff{eq:feedback-QRST-QSI-1}-\reff{eq:feedback-QRST-QSI-2} is achievable.
Furthermore, there is no need to double-block the protocol as above because our
converse theorem for this case demonstrates that it is not necessary to do so.

\begin{figure}[ptb]
\begin{center}
\includegraphics[
width=3.5in
]{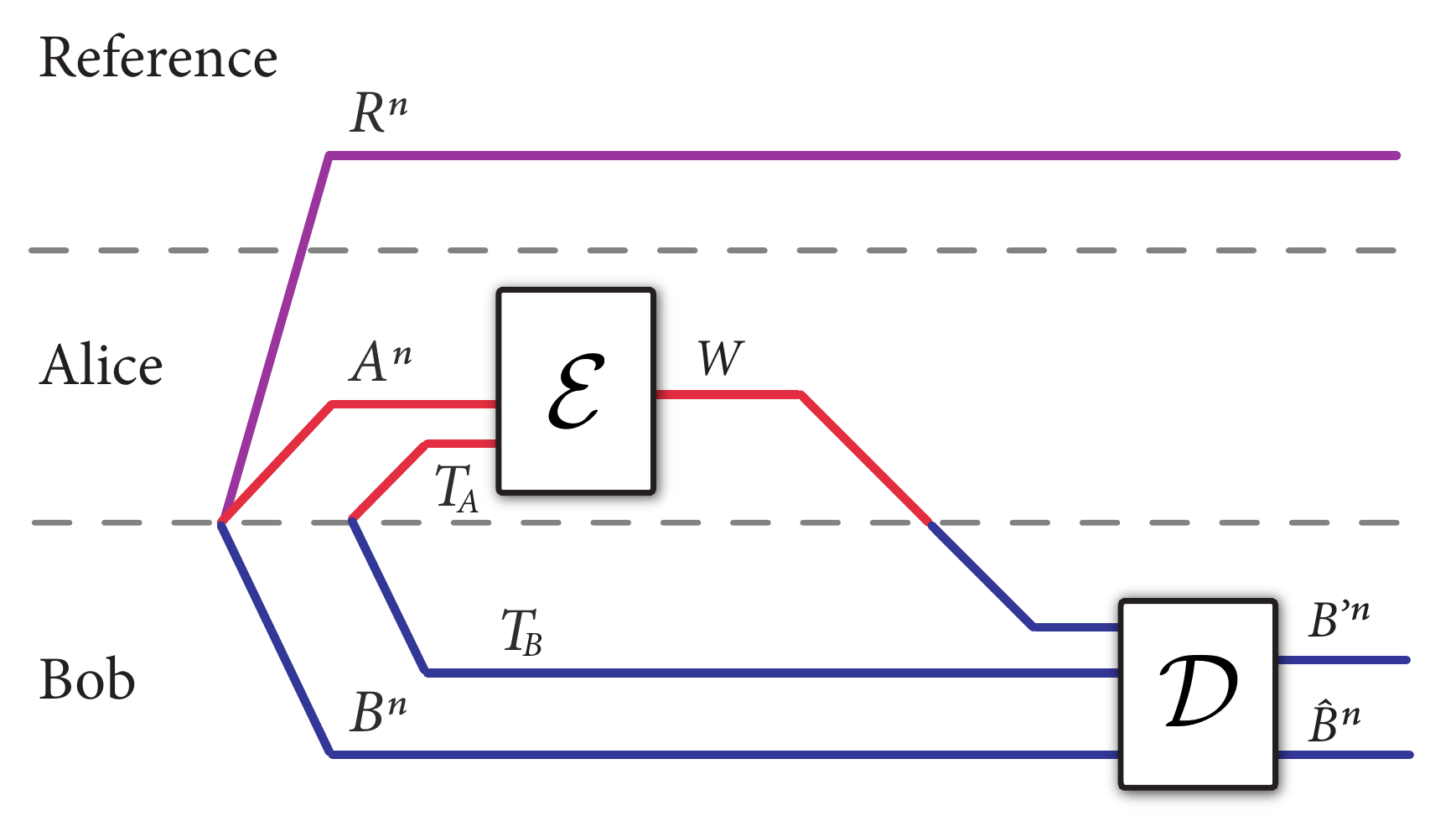}
\end{center}
\caption{A general protocol for simulating a quantum channel when quantum side information and shared
entanglement are available.}
\label{fig:QRST-QSI}%
\end{figure}

We now prove the converse part of Theorem~\ref{thm:QRST-QSI-non-feedback}, that is, we
establish that the bounds \reff{eq:non-feedback-QRST-QSI-1}-\reff{eq:non-feedback-QRST-QSI-2} 
hold for any protocol that results in a non-feedback simulation of 
the noisy channel $\cN_{A \to B'}$. The
most general protocol begins with the reference, Alice, and Bob sharing $n$
copies of $\phi_{RAB}^{\rho}$, 
where $n$ is some arbitrarily large positive
integer. Alice and Bob also share some entangled state $\Phi_{T_{A}T_{B}}$\ on
systems $T_{A}$ and $T_{B}$, which they can use to help them in their task. In
particular, the Schmidt rank of this entangled state is equal to $2^{nE}$, so
that $E$ quantifies the rate of entanglement consumption. Alice performs some
encoding map $\mathcal{E}$\ on systems $A^{n}$ and $T_{A}$ which produces a
system $W$ as output. She sends system $W$ to Bob, who then performs a
decoding map $\mathcal{D}$\ on systems $W$, $T_{B}$, and $B^{n}$. This
decoding map has two outputs $B^{\prime n}$ and $\hat{B}^{n}$, where
$B^{\prime n}$ approximates the output of the channel simulation and $\hat
{B}^{n}$ represents an approximation of Bob's quantum side information. If the
protocol is any good for accomplishing the task of a non-feedback channel
simulation, then the output of this simulation protocol and the output of the
ideal protocol should be asymptotically indistinguishable as $n$ becomes
large. That is, for any arbitrary $\eps >0$, for $n$ large enough, it should hold that
\begin{multline}
\bigg\Vert \mathcal{D}_{WT_{B}B^{n}\rightarrow B^{\prime n}\hat{B}^{n}}\left(
\mathcal{E}_{A^{n}T_{A}\rightarrow W}\left(  \left(  \phi_{RAB}^{\rho}\right)
^{\otimes n}\otimes\Phi_{T_{A}T_{B}}\right)  \right)  \\ -\left(  \mathcal{N}%
_{A\rightarrow B^{\prime}}\left(  \phi_{RAB}^{\rho}\right)  \right)  ^{\otimes
n}\bigg\Vert _{1}\leq\eps.
\label{eq:nonfeedback-QRST-QSI-good-simulation-cond}%
\end{multline}
Figure~\ref{fig:QRST-QSI} illustrates the protocol discussed above.

A useful observation for proving the converse is that an arbitrary
encoding map $\mathcal{E}$ and a decoding map $\mathcal{D}$ can be simulated by a
protocol involving only isometric operations. In particular, we can replace
the encoding $\mathcal{E}$ with an isometric extension $U^{\mathcal{E}}$\ that
has as output the original output~$W$ and an environment system $E_{1}$. Let
$\sigma\equiv \sigma_{RE_{1}WT_{B}B^{n}}$ denote the pure state 
shared between the reference, Alice, and Bob after the
action of $U^{\mathcal{E}}$. We can also replace the decoding map 
$\mathcal{D}$ with an isometric extension of it that has as output the
original outputs $B^{\prime n}\hat{B}^{n}$ and an additional environment
system~$E_{2}$. Let $\omega_{RE_{1}E_{2}B^{\prime n}\hat{B}^{n}}$ denote the
pure state resulting from applying an isometric extension of the decoder to
the state $\sigma$.

Due to monotonicity of the trace distance under partial trace,
the condition in
(\ref{eq:nonfeedback-QRST-QSI-good-simulation-cond}) implies that the
inequality in (\ref{trace}) holds,%
\begin{figure*}
\be\label{trace}
\left\Vert \text{Tr}_{E_{1}E_{2}B^{\prime n}}\left\{  U^{\mathcal{D}}%
_{WT_{B}B^{n}\rightarrow E_{2}B^{\prime n}\hat{B}^{n}}\left(  U^{\mathcal{E}%
}_{A^{n}T_{A}\rightarrow WE_{1}}\left(  \left(  \phi_{RAB}^{\rho}\right)
^{\otimes n}\otimes\Phi^{T_{A}T_{B}}\right)  \right)  \right\}  -\left(
\phi_{RB}^{\rho}\right)  ^{\otimes n}\right\Vert _{1}\leq\eps, 
\ee
\end{figure*}
with the partial trace of the states in \reff{eq:nonfeedback-QRST-QSI-good-simulation-cond} being taken over the system $B^{\prime n}$.

Since $\omega_{RE_{1}E_{2}B^{\prime n}\hat{B}^{n}}$ is a purification of the
first term in the trace distance in \reff{trace}, and $\sigma_{RE_{1}WT_{B}B^{n}}$ is a
purification of the second term in the trace distance, Uhlmann's theorem \cite{U73,J94} implies the existence of an isometry $V\equiv V_{E_{1}WT_{B}\rightarrow E_{1}%
E_{2}B^{\prime_{n}}}$ such that the trace distance between $V(\sigma_{RE_{1}WT_{B}B^{n}})$ and
$\omega_{RE_{1}E_{2}B^{\prime n}\hat{B}^{n}}$ is no larger than $2\sqrt
{\epsilon}$. Let $\omega^{\prime}$ denote the state resulting from applying
$V$ to $\sigma_{RE_{1}%
WT_{B}B^{n}}$. Thus, the original decoder can be simulated by the isometry
$V$ which does not act on
Bob's quantum side information (we should expect for this to be possible,
given that the condition in
(\ref{eq:nonfeedback-QRST-QSI-good-simulation-cond}) implies that Bob's
QSI\ should not be disturbed too much). We also observe that $\omega^{\prime}$
is $\eps$-close in trace distance to a state of the form in
(\ref{eq:code-state-QRST-QSI}) because $U^{\mathcal{E}}$ and $V$ do not act
on the systems $R^n B^n$. By applying Uhlmann's theorem once again and
the triangle inequality, we conclude that there exists some isometry
$U\equiv U_{E^{n}\rightarrow E_{1}E_{2}}$ such that when $U^{\dag}$\ is applied to
$\omega^{\prime}$, the resulting state is close in trace distance to $\left(
U^{\mathcal{N}}_{A\rightarrow B^{\prime}E}\left(  \phi_{RAB}^{\rho}\right)
\right)  ^{\otimes n}$. Figure~\ref{fig:QRST-QSI-uhlmann} summarizes the
observations made in this paragraph and the previous one.

\begin{figure*}[ptb]
\begin{center}
\includegraphics[
width=6.6in
]{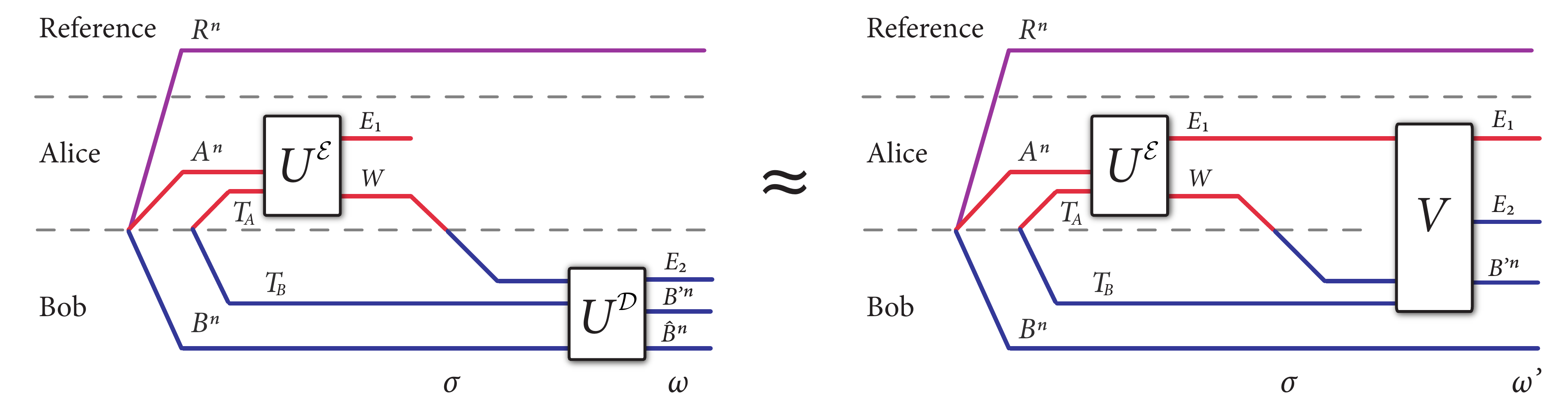}
\end{center}
\caption{The protocol on the left is a simulation of the general protocol in Figure~\ref{fig:QRST-QSI},
in which Alice acts with an isometric extension~$U^{\mathcal{E}}$ of her encoder~$\mathcal{E}$ and Bob acts with an isometric extension~$U^\mathcal{D}$ of his decoder~$\mathcal{D}$. Due to the fact that the simulation causes only
a negligible disturbance to the state of $R^n B^n$, by Uhlmann's theorem, there exists an isometry $V_{E_{1}WT_{B}\rightarrow E_{1}E_{2}B^{\prime_{n}}}$ such that the final state on the left is approximately equal to the final state on the right.}
\label{fig:QRST-QSI-uhlmann}%
\end{figure*}

%
%
Consider the following state:%
\begin{equation}
\tau\equiv U_{E^{n}\rightarrow E_{1}E_{2}}\left(  U^{\mathcal{N}%
}_{A\rightarrow B^{\prime}E}\left(  \phi_{RAB}^{\rho}\right)  \right)
^{\otimes n}, \label{eq:tau-QRST-QSI-1}%
\end{equation}
where $ \phi_{RAB}^{\rho}$ is a purification of $\rho_{AB}$, $U^{\cN}_{A\to B'E}$ is a Stinespring isometry of the noisy channel $\cN$ which is to be simulated, and $U_{E^{n}\rightarrow E_{1}E_{2}}$ is an isometry. This state is of the form as given in (\ref{eq:code-state-QRST-QSI}).
We proceed to find a lower bound on the
optimal rate $Q$ of quantum communication needed for the channel simulation as
follows:%
\begin{align}
nQ \equiv  \log \left({\rm{dim}} \cH_W\right) &  \geq H\left(  W\right)  _{\sigma}\nonumber\\
&  =H\left(  WT_{B}B^{n}\right)  _{\sigma}-H\left(  B^{n}T_{B}|W\right)
_{\sigma}\nonumber\\
&  =H(E_{2}B^{\prime n}\hat{B}^{n})_{\omega}-H\left(  B^{n}T_{B}|W\right)
_{\sigma}\nonumber\\
&  \geq H(E_{2}B^{\prime n}B^{n})_{\omega^{\prime}}-n\eps^{\prime
}-H\left(  B^{n}\right)  _{\sigma}\nonumber \\
& \,\,\,\,-H\left(  T_{B}\right)  _{\sigma}\nonumber\\
&  =H(E_{2}B^{\prime n}B^{n})_{\omega^{\prime}}-n\eps^{\prime}-H\left(
B^{n}\right)  _{\omega^{\prime}} \nonumber \\
& \,\,\,\, -H\left(  T_{B}\right)  _{\sigma}\nonumber\\
&  \geq H(E_{2}B^{\prime n}|B^{n})_{\omega^{\prime}}-nE-n\eps^{\prime
}\nonumber\\
&  \geq H(E_{2}B^{\prime n}|B^{n})_{\tau}-nE-2n\eps^{\prime}%
.\label{eq:1st-bound-QRST-QSI-non-f}%
\end{align}%
The first equality is an entropy
identity. The second equality follows because entropy is invariant under the
action of an isometry (in this case, the isometry is $U^{\mathcal{D}}$). The
second inequality follows from Uhlmann's theorem (as mentioned above), the
Alicki-Fannes' inequality (continuity of entropy)\ with an appropriate choice
of $\eps^{\prime}$ (a similar convention to what we had in our previous
converse theorems), and from subadditivity of entropy:%
\[
H\left(  B^{n}T_{B}|W\right)  _{\sigma}\leq H\left(  B^{n}T_{B}\right)
_{\sigma}\leq H\left(  B^{n}\right)  _{\sigma}+H\left(  T_{B}\right)
_{\sigma}.
\]
The third equality follows because the map $V^{E_{1}WT_{B}\rightarrow
E_{1}E_{2}B^{\prime_{n}}}$ does not act on the $B^{n}$ system. The third
inequality follows because $H\left(  T_{B}\right)  _{\sigma}\leq nE$. The
final inequality follows from another application of the Alicki-Fannes'
inequality to the state $\tau$ defined in (\ref{eq:tau-QRST-QSI-1}).

We prove the second bound in (\ref{eq:non-feedback-QRST-QSI-2}):%
\begin{align}
I\left(  R^{n};B^{\prime n}B^{n}E_{2}\right)  _{\tau}  & \leq I\left(
R^{n};B^{\prime n}B^{n}E_{2}\right)  _{\omega}+n2\varepsilon^{\prime
}\nonumber\\
& = I\left(  R^{n};B^{n}WT_{B}\right)  _{\sigma}+n2\varepsilon^{\prime
}\nonumber\\
& \leq n2Q+I\left(  R^{n};B^{n}T_{B}\right)  _{\sigma}+n2\varepsilon^{\prime
}\nonumber\\
& =n2Q+I\left(  R^{n};B^{n}\right)  _{\sigma}+n2\varepsilon^{\prime
}\nonumber\\
& =n2Q+I\left(  R^{n};B^{n}\right)  _{\tau}+n2\varepsilon^{\prime
}\label{eq:chain-for-second-bound}%
\end{align}
The first inequality results from the fact that $\tau$ is $\varepsilon$-close
to $\omega^{\prime}$, $\omega^{\prime}$ is $\varepsilon$-close to $\omega$,
and from applying the Alicki-Fannes inequality. The second inequality follows
from quantum data processing. The third inequality follows from the following
property of quantum conditional mutual information:%
\begin{align*}
I\left(  A;BC\right)    & =I\left(  A;B\right)  +I\left(  A;C|B\right)  \\
& =I\left(  A;B\right)  +H(A|B)-H\left(  A|BC\right)  \\
& \leq I\left(  A;B\right)  +H\left(  A\right)  +H\left(  A|D\right)  \\
& \leq I\left(  A;B\right)  +2\log\left\vert A\right\vert ,
\end{align*}
where the entropies are evaluated on a tripartite state $\rho_{ABC}$ purified
by some state on a purifying system $D$. The first equality in the chain in
(\ref{eq:chain-for-second-bound}) follows because $T_{B}$ is product with
systems $R^{n}$ and $B^{n}$ for the state $\sigma$. The above inequalities
then imply that%
\[
n2Q\geq I\left(  R^{n};B^{\prime n}E_{2}|B^{n}\right)  _{\tau}-n2\varepsilon
^{\prime},
\]
by finally applying the chain rule for quantum mutual information.

To obtain a converse proof for Theorem~\ref{thm:QRST-QSI-feedback}, we can
exploit the above converse with just one further observation. Since a feedback
simulation requires Alice to possess the full environment of the simulation,
the final state on $R^{n}E_{1}B^{\prime n}B^{n}$ must be a
pure state and $E_{1}$ must be unitarily related to the $E^{n}$ system of
$\left(  U^{\mathcal{N}}_{A\rightarrow B^{\prime}E}\left(  \phi_{RAB}^{\rho
}\right)  \right)  ^{\otimes n}$. Thus, it must be the case that the system
$E_{2}$ is product with $R^{n}E_{1}B^{\prime n}B^{n}$. This final observation
leads to a single-letterization of the above bounds as follows:%
\begin{align*}
nQ  &  \geq H(E_{2}B^{\prime n}|B^{n})_{\tau}-nE-2n\eps^{\prime}\\
&  =H(B^{\prime n}|B^{n})_{\tau}-nE-2n\eps^{\prime}\\
&  =n\left[  H(B^{\prime}|B)_{\mathcal{N}_{A\rightarrow B^{\prime}}\left(
\rho_{AB}\right)  }-E-2\eps^{\prime}\right]  ,
\end{align*}
and similarly,%
\begin{align*}
n2Q  &  \geq I\left(  E_{2}B^{\prime n};R^{n}|B^{n}\right)  _{\tau}%
-2n\eps^{\prime}\\
&  =I\left(  B^{\prime n};R^{n}|B^{n}\right)  _{\tau}-2n\eps^{\prime}\\
&  =n\left[  I\left(  B^{\prime};R|B\right)  _{\mathcal{N}_{A\rightarrow
B^{\prime}}\left(  \psi_{RAB}^{\rho}\right)  }-2\eps^{\prime}\right]  .
\end{align*}
The main step in the above equalities is to exploit the observation that
$E_{2}$ must be product with the other systems for a feedback simulation.
\end{proof}

\subsection{On a General Quantum Reverse Shannon Theorem with QSI}
With the above tensor-power quantum reverse Shannon theorem with QSI in hand,
one might be tempted to pursue a general form of this theorem that holds whenever the input
to the channel is a general, non-IID state entangled with a system available at the receiver's end
(non-IID input and quantum side information). From the above theorem and the techniques developed in
Refs.~\cite{BDHSW09,BCR09}, we suspect that it is possible
to show that the following rate of quantum communication
is necessary and sufficient for simulating an IID channel acting on a general input entangled with
a system at the receiving end, whenever unlimited entanglement in any form is available between the
sender and receiver:
\begin{equation}
\tfrac12 \max_{\phi_{RAB}}I\left(  R;B^{\prime
}|B\right)  _{\mathcal{N}_{A\rightarrow B^{\prime}}\left(  \phi\right)  } . \label{eq:conjectured-QRST-QSI-rate}
\end{equation}
Also, it is known from Refs.~\cite{BDHSW09,BCR09} that the following rate of quantum communication
is necessary and sufficient for simulating an IID channel acting on a general input (neglecting any quantum side information), whenever unlimited entanglement in any form is available between the
sender and receiver:
$$
\tfrac12 \max_{\psi_{RA}}I\left(  R;B^{\prime}\right)  _{\mathcal{N}_{A\rightarrow
B^{\prime}}\left(  \psi\right)  }.
$$
The following theorem clarifies that these two rates are in fact equal,
implying that pursuing a general quantum reverse Shannon theorem with QSI is a pointless task if the conjectured rate in \eqref{eq:conjectured-QRST-QSI-rate} is correct
(at least for a feedback simulation). The reason that such a relation should hold
is that the theorems from Refs.~\cite{BDHSW09,BCR09} are simulating an IID channel with respect to the diamond norm,\footnote{See, e.g., Refs.~\cite{BDHSW09,BCR09} for a definition of the diamond norm.} 
which is known to be robust under tensoring with other systems on which the channel does not act.
\begin{theorem}
The following identity holds%
\[
\max_{\psi_{RA}}I\left(  R;B^{\prime}\right)  _{\mathcal{N}_{A\rightarrow
B^{\prime}}\left(  \psi\right)  }=\max_{\phi_{RAB}}I\left(  R;B^{\prime
}|B\right)  _{\mathcal{N}_{A\rightarrow B^{\prime}}\left(  \phi\right)  },
\]
where each maximization is over pure states.
\end{theorem}

\begin{proof}
We first prove the following inequality:%
\[
\max_{\psi_{RA}}I\left(  R;B^{\prime}\right)  _{\mathcal{N}_{A\rightarrow
B^{\prime}}\left(  \psi\right)  }\leq\max_{\phi_{RAB}}I\left(  R;B^{\prime
}|B\right)  _{\mathcal{N}_{A\rightarrow B^{\prime}}\left(  \phi\right)  }.
\]
Let $\psi_{RA}^{\ast}$ be the state that achieves the maximum of the LHS. Then
the state $\psi_{RA}^{\ast}\otimes\varphi_{B}$ (for any pure state
$\varphi_{B}$) leads to $\mathcal{N}_{A\rightarrow B^{\prime}}\left(
\psi_{RA}^{\ast}\right)  \otimes\varphi_{B}$ at the channel output, and it is
a special pure state included in the maximization on the RHS, so that%
\begin{align*}
I\left(  R;B^{\prime}\right)  _{\mathcal{N}_{A\rightarrow B^{\prime}}\left(
\psi^{\ast}\right)  }  & =I\left(  R;B^{\prime}|B\right)  _{\mathcal{N}%
_{A\rightarrow B^{\prime}}\left(  \psi^{\ast}\otimes\varphi\right)  }\\
& \leq\max_{\phi_{RAB}}I\left(  R;B^{\prime}|B\right)  _{\mathcal{N}%
_{A\rightarrow B^{\prime}}\left(  \phi\right)  }.
\end{align*}
We now prove the other inequality:%
\begin{equation}
\max_{\psi_{RA}}I\left(  R;B^{\prime}\right)  _{\mathcal{N}_{A\rightarrow
B^{\prime}}\left(  \psi\right)  }\geq\max_{\phi_{RAB}}I\left(  R;B^{\prime
}|B\right)  _{\mathcal{N}_{A\rightarrow B^{\prime}}\left(  \phi\right)
}.\label{eq:other-maximization}%
\end{equation}
Let $\phi_{RAB}^{\ast}$ be the pure state that achieves the maximum on the
RHS. Then we have that%
\begin{align*}
I\left(  R;B^{\prime}|B\right)  _{\mathcal{N}_{A\rightarrow B^{\prime}}\left(
\phi^{\ast}\right)  }  & =I\left(  RB;B^{\prime}\right)  _{\mathcal{N}%
_{A\rightarrow B^{\prime}}\left(  \phi^{\ast}\right)  } \\
& \,\,\, \,\,-I\left(  B;B^{\prime
}\right)  _{\mathcal{N}_{A\rightarrow B^{\prime}}\left(  \phi^{\ast}\right)
}\\
& \leq I\left(  RB;B^{\prime}\right)  _{\mathcal{N}_{A\rightarrow B^{\prime}%
}\left(  \phi^{\ast}\right)  }\\
& \leq\max_{\psi_{RA}}I\left(  R;B^{\prime}\right)  _{\mathcal{N}%
_{A\rightarrow B^{\prime}}\left(  \psi\right)  }.
\end{align*}
The first equality follows from the chain rule for quantum mutual information,
and the first inequality follows because $I\left(  B;B^{\prime}\right)  \geq
0$. The final inequality follows because $\phi^{\ast}$ is some pure bipartite
state with respect to the cut $RB|A$, where the systems $R$ and $B$
purify the input to the channel, and this is of course a particular kind of
state in the maximization on the LHS\ of (\ref{eq:other-maximization}).
\end{proof}

\subsection{Quantum Rate Distortion with Quantum Side Information}
\label{qsi}

Having established the quantum reverse Shannon theorem with QSI, we now prove
a theorem characterizing the quantum rate distortion function
in the presence of QSI, which is denoted as  $R^q_{qsi}(D)$ and was introduced 
at the beginning of Section~\ref{sec-qrd-qsi}.
%

\begin{theorem}
\label{thm:QSI-QRD-no-disturb} Consider a bipartite state $\rho_{AB}$, obtained
by the action of an isometry on the source state of a memoryless quantum 
information source. Suppose Alice has the system $A$ and Bob has the system $B$, the latter acting as QSI.
Let $\phi^\rho_{RAB}$ be a purification of $\rho_{AB}$. Then 
for any given distortion $D\geq0$, 
the quantum rate distortion function with QSI, evaluated under the condition 
that the protocol causes only a negligible
disturbance to the joint state of the $BR$ systems, is given by%
\begin{equation}
R^{q}_{qsi}\left(  D\right)  =\lim_{k\rightarrow\infty}\frac{1}{k}%
\min_{\mathcal{N}^{\left(  k\right)  }\ :\ \overline{d}\left(  \rho,\,\mathcal{N}^{\left(  k\right)  }\right)  \leq D}I_{p}\left(  \left(
\rho_{AB}\right)  ^{\otimes k},\cN^{(k)}\right) \label{eq:QRD-QSI}%
\end{equation}
where%
\begin{align*}
I_{p}\left(  \rho_{AB},\ \mathcal{N}_{A\rightarrow B^{\prime}}\right)   &
\equiv \inf_{V_{E\rightarrow E_{A}E_{B}}}
\max\bigg\{ \frac{1}{2}I\left(R;B^{\prime}E_{B}|B\right)_{\omega},\\
& \ H\left(  B^{\prime}E_{B}|B\right)_{\omega}\bigg\},
\end{align*}
with
$\omega_{RE_{A}E_{B}B^{\prime}B}  \equiv V_{E\rightarrow E_{A}E_{B}}
\left(\ U^{\mathcal{N}}_{A\rightarrow B^{\prime}E}\left( \phi^{\rho}_{RAB}\right)\right)$. 
In the above, 
$\cN_{A \to B'}$ and $\cN^{(k)}: \cD(\cH_{A^k}) \to \cD(\cH_{B^k})$ are 
CPTP maps, $U^{\mathcal{N}}_{A\rightarrow B^{\prime}E}$ is an isometric extension
of $\cN_{A \to B'}$, and $V_{E\rightarrow E_{A}E_{B}}$ is an environment-splitting isometry.
\end{theorem}

\begin{proof}
The achievability part of this theorem follows easily from
Theorem~\ref{thm:QRST-QSI-non-feedback}. In particular, we just fix the map
that achieves the minimum in (\ref{eq:QRD-QSI}), and it easily follows that
performing the protocol from Theorem~\ref{thm:QRST-QSI-non-feedback} leads to
a state on $R^{n}$ and $B^{\prime n}$ that has distortion no larger than $D$
with the original state on $R^{n}$ and $A^{n}$. We can block the protocol to
achieve the regularized formula.

We now prove the converse. As before, the most general protocol begins with the
reference, Alice, and Bob sharing $n$ copies of the state $\phi^{\rho}_{RAB}$.
Alice performs some
encoding map $\mathcal{E}$\ on the systems $A^{n}$, obtaining a quantum system
$W$. She sends system$~W$ to Bob using noiseless qubit channels. Bob 
feeds this system, and his quantum side information $B^{n}$, into a
decoding map $\mathcal{D}$, which produces as output systems $B^{\prime n}$
and $\hat{B}^{n}$. We demand that the distortion of the state of systems
$R^{n}B^{\prime n}$ with respect to the state of $R^{n}A^{n}$ at the start of the protocol be no larger than~$D$. Also, we demand that the decoder causes only
an asymptotically negligible disturbance of the joint state of 
the reference and the quantum side
information, in the sense that for any $\eps >0$, for $n$ large enough:%
\begin{multline}
\bigg\Vert \text{Tr}_{B^{\prime n}}\left\{  \mathcal{D}^{WB^{n}\rightarrow
B^{\prime n}\hat{B}^{n}}\left(  \mathcal{E}^{A^{n}\rightarrow W}\left(
\left(  \phi^{\rho}_{RAB}\right)  ^{\otimes n}\right)  \right)  \right\}
\\-\left(  \phi^{\rho}_{RB}\right)  ^{\otimes n}\bigg\Vert _{1}\leq
\eps.\label{eq:do-not-disturb}%
\end{multline}

Like our other converses, the key to this proof is the realization that the above general
protocol can be simulated by one in which we exploit an isometric extension of
the encoder $U^{\mathcal{E}}$, which maps $A^{n}$ to $W$ and an environment
$E_{1}$. Let $\sigma$ denote the overall state after $U^{\mathcal{E}}$ acts
(so that $\sigma_{R^{n}E_{1}WB^{n}}$ is a pure state). We also exploit an
isometric extension $U^{\mathcal{D}}$ of the decoder, which maps $WB^{n}$ to
$B^{\prime n}\hat{B}^{n}$ and an environment $E_{2}$. Let $\omega\equiv\omega_{R^{n}E_{1}E_{2}B^{\prime n}\hat{B}^{n}}$ denote the
overall state after $U^{\mathcal{D}}$ acts. We proceed with bounding the rate $Q$ of
noiseless quantum communication by following the same steps as in
(\ref{eq:1st-bound-QRST-QSI-non-f}) and (\ref{eq:chain-for-second-bound}),
but ignoring the entanglement assistance:%
\bea
nQ &  \geq & H(E_{2}B^{\prime n}|B^{n})_{\omega^{\prime}}-n\epsilon^{\prime},\label{aa}\\
nQ &  \geq &\frac{1}{2}I\left(  E_{2}B^{\prime n};R^{n}|B^{n}\right)
_{\omega^{\prime}}-n\epsilon^{\prime}.%
\label{bb}
\eea
Putting things together, we obtain the following lower bound on the quantum
communication rate:%
\begin{align*}
nQ &  \geq\max\left\{  \frac{1}{2}I\left(  E_{2}B^{\prime n};R^{n}%
|B^{n}\right)  _{\omega^{\prime}},\ H(E_{2}B^{\prime n}|B^{n})_{\omega
^{\prime}}\right\}  -n\eps^{\prime}\\
&  \geq
\inf_{V^{E^{n}\rightarrow E_{1}E_{2}}}\max\bigg\{  \frac{1}{2}I\left(
E_{2}B^{\prime n};R^{n}|B^{n}\right)  _{\omega^{\prime}},\\
& \,\,\,\,\,\ \ \ \ \ \ \ \ \ \ \ \ H(E_{2}B^{\prime
n}|B^{n})_{\omega^{\prime}}\bigg\}  -n\eps^{\prime}\\
&  \geq\min_{\mathcal{N}^{\left(  n\right)  }\ :\ d\left(  \rho^{\otimes
n},\mathcal{N}^{\left(  n\right)  }\right)  \leq D}I_{p}\left(  \left(
\rho_{AB}\right)  ^{\otimes n},\mathcal{N}^{\left(  n\right)
}\right)  -n\eps^{\prime,}%
\end{align*}
where $\cN^{(n)}: \cD(\cH_{A^n}) \to \cD(\cH_{B^n})$ denotes a CPTP map.
The first inequality follows by combining \reff{aa} and \reff{bb}. The second
inequality follows by performing an optimization over all possible splits of the
environment. The final inequality results from minimizing $I_{p}$ over all
maps $\mathcal{N}^{\left(  n\right)  }$ that act only on $A^{n}$ and meet the
distortion criterion.
\end{proof}


\subsection{Entanglement-Assisted Quantum Rate Distortion with Quantum Side
Information}
\label{ent-ass-qsi}

If in addition to Bob having quantum side information, Alice and Bob have
prior shared entanglement (over systems $T_A$ and $T_B$) as an auxiliary resource, then the corresponding rate 
distortion function, for any given distortion $D$ is denoted as $R^{q}_{ea,qsi}(D)$.
It is defined analogously as in \reff{ea-enc}-\reff{ea-dec}, except that Bob can use the QSI in his decompression
task. Hence the decompression map is given by
$$\mathcal{D}_{n} : \mathcal{D}({\widetilde{\mathcal{H}}_{Q^{n}}} \otimes  \cH_{T_B} \otimes \cH_{B^n})\rightarrow
\mathcal{D}(\mathcal{H}_{A}^{\otimes n}).
$$
We prove the following theorem:

\begin{theorem}
\label{thm:EA-QSI-QRD}

Suppose Alice and Bob share entanglement and a state $\rho_{AB}$
(obtained from a memoryless quantum information source), such that the system
$A$ is with Alice and $B$ is with Bob. Let $\phi^{\rho}_{RAB}$ be
a purification of $\rho_{AB}$. Then the quantum
rate distortion function, $R^{q}_{ea,qsi}(D)$,
evaluated under the condition that the protocol
causes only a negligible disturbance to the joint state of $BR$ is given by%
\begin{equation}
R^{q}_{ea,qsi}\left(  D\right)  =\frac{1}{2}\left[\min_{\mathcal{N}\ :\ d\left(
\rho,\mathcal{N}\right)  \leq D}I\left( R;B^{\prime}|B\right)  _{\sigma}\right]
\label{eq:EA-QSI-QRD}%
\end{equation}
where the state $\sigma$ is defined as:%
\[
\sigma_{RB^{\prime}B}\equiv\mathcal{N}_{A\rightarrow B^{\prime}}\left(
\phi^{\rho}_{RAB}\right).
\]
\end{theorem}

\begin{proof}
The achievability part follows easily from the protocol for quantum state
redistribution. Fix $\mathcal{N}$ to be the CPTP map which achieves the
minimum in (\ref{eq:EA-QSI-QRD}) for a given distortion $D$. It is
straightforward to compute this minimum because the information quantity
$I\left(  B^{\prime};R|B\right)  _{\sigma}$ is convex in the map $\mathcal{N}$.
This follows easily from the identity $I\left(  B^{\prime};R|B\right)  _{\sigma
}=I\left(  B^{\prime}B;R\right)  _{\sigma}-I\left(  B;R\right)  _{\sigma}$, and
the fact that the map, $\cN$, acts only on the system $A$. 
The reference, Alice, and Bob share $n$
copies of $\phi_{RAB}^{\rho}$.
First Alice acts on her system $A^n$ with many instances of an isometric extension
$U^{\mathcal{N}}_{A\rightarrow B^{\prime}E}$ of the map 
$\mathcal{N}_{A\rightarrow B^{\prime}}$. Let 
$\left(\phi^{\sigma}_{RB^{\prime}EB}\right)^{\otimes n}$ denote the
resulting pure state, with the systems $B^{\prime n}E^n$ being in Alice's possession. Note that since $\cN$ is chosen to be the CPTP map which meets the distortion constraint $d(\rho, \cN) \le D$, the rate-distortion task is completed if the systems $B^{\prime n}$ are transmitted to Bob faithfully in the asymptotic limit ($n \to \infty$). This is accomplished by using the protocol of quantum state redistribution. From \reff{resource-ent} it follows that the
relevant resource inequality is given by:%
\begin{multline}
\langle\phi^{\sigma}_{R|B^{\prime}E|B}\rangle+\tfrac{1}{2}I\left( R; B^{\prime
}|B\right)  _{\phi^{\sigma}}\left[  q\rightarrow q\right]  +\tfrac{1}%
{2}I\left( E; B^{\prime}\right)  _{\phi^{\sigma}}\left[  qq\right]
\\
\geq\langle\phi^{\sigma}_{R|E|B^{\prime}B}\rangle+\tfrac{1}{2}I\left(
B;B^{\prime}\right)  _{\phi^{\sigma}}\left[  qq\right]  .
\end{multline}
Clearly, with unlimited entanglement, the
protocol accomplishes the state redistribution task with a rate of quantum communication given by \reff{eq:EA-QSI-QRD}.

We now prove the converse part of the above theorem. This proof bears some
similarities with our converse proof of Theorem~3~of Ref.~\cite{DHW11} and
with the converse proof of Theorem~6 of \cite{DHWW12}. The most general protocol begins with the reference, Alice, and Bob
sharing the state $\left(  \phi^{\rho}_{RAB}\right)  ^{\otimes n}$, and Alice
and Bob sharing entanglement in the systems $T_{A}$ and $T_{B}$, respectively
(such that the dimensions of $\cH_{T_A}$ and $\cH_{T_B}$ are no larger than $2^{nE}$).
Alice then acts with some encoding map $\mathcal{E}$\ on $A^{n}$ and $T_{B}$,
producing a system $W$. Let $\sigma$ denote the state shared by Alice, Bob and the reference after the encoding. She
sends $W$ to Bob, who then acts with a decoding map $\mathcal{D}$\ on $W$, his
share $T_{B}$ of the entanglement, and his quantum side information $B^{n}$ to
produce systems $B^{\prime n}$ and $\hat{B}^{n}$. Let $\omega$ denote the
final shared state. Without loss of generality, we can simulate the above protocol by considering an
isometric extension of the encoder that produces systems $W$ and $E$ as
output. We demand that the protocol causes only an asymptotically negligible
disturbance to the state on the reference and Bob's systems, in the sense that,
for any $\eps >0$ and $n$ large enough, the inequality in (\ref{trace2}) should hold.%
\begin{figure*}
\be\label{trace2}
\left\Vert \text{Tr}_{B^{\prime n}}\left\{  \mathcal{D}_{W T_B B^n\to B'^n\hat{B}^n}\left(  \mathcal{E}_{A^nT_A\to W}%
\left(  \left(  \phi^{\rho}_{RAB}\right)  ^{\otimes n}\otimes\Phi_{T_{A}T_{B}%
}\right)  \right)  \right\}  -\left(  \phi^{\rho}_{RB}\right)  ^{\otimes
n}\right\Vert _{1}\leq\eps,
\ee
\end{figure*}
Then the rate of quantum communication $Q\equiv(1/n) \log\left({\rm{dim }}\cH_W\right)$, needed for the rate-distortion coding task, satisfies the following bound:%
\be
2nQ \geq nR^{q}_{ea,qsi}\left(  D\right)  -3n\eps^{\prime},
\ee
for any $\eps^{\prime} >0$ and $n$ large enough. 
It is proved by employing standard entropic identities and inequalities, e.g. the quantum-data processing inequality, the Alicki-Fannes inequality \cite{AF04}, and the superadditivity of the quantum mutual information (Lemma~\ref{super}). It also relies on the fact that the state $R^{n}\hat{B}^{n}$ is $\eps$-close in trace distance to a
tensor-product state. For the sake of completemess, we have included the proof in Appendix~\ref{app-ea-converse}.
\end{proof}

\section{Conclusion}

\label{sec:conclusion}
We have extended quantum rate distortion theory by considering auxiliary resources
that might be available to the sender and receiver. The first setting we considered
is quantum rate distortion coding with the help of a classical side channel. Our result
is that the regularized entanglement of formation characterizes the quantum rate
distortion function, extending earlier work of Devetak and Berger \cite{Devetak:2002it}.
We also combined this bound with our entanglement-assisted bound from Ref.~\cite{DHW11}
to obtain the best known bounds on the quantum rate distortion function for an isotropic
qubit source. The second setting we considered is quantum rate distortion coding with quantum side
information available to the receiver. Before proving results in this setting, we proved
a quantum reverse Shannon theorem with quantum side information (for
tensor-power input states), which naturally extends the 
quantum reverse Shannon theorem (for tensor-power inputs) in Ref.~\cite{BDHSW09}. The achievability part of this
theorem relies on the quantum state redistribution protocol \cite{DY08,YD09}, while the converse
relies on the fact that the protocol can cause only a negligible disturbance to the state of the
reference and Bob's quantum side information. This result naturally leads to quantum rate-distortion theorems
with quantum side information, with or without entanglement assistance.

All of our proofs rely on one particular approach to quantum rate distortion theory: exploiting
a quantum reverse Shannon theorem for the task of quantum rate distortion coding. It would be
a breakthrough for this theory if one could develop a different approach that leads to better
characterizations of lossy quantum data compression tasks, beyond the ones presented here.

\section*{Acknowledgements}

The authors are grateful to Mario Berta, Patrick Hayden, and Ke Li for useful discussions,
and to Jianxin Chen for his help with generating the entanglement of purification 
plot in Figure~\ref{fig:RD-compare}. We are also grateful for feedback of the
anonymous referees.

MMW\ acknowledges support
from the Centre de Recherches Math\'{e}matiques at the University of Montreal.
MH received support from the Chancellor's postdoctoral research fellowship,
University of Technology Sydney (UTS), and was also partly supported by the
National Natural Science Foundation of China (Grant No.~61179030) and the
Australian Research Council (Grant No.~DP120103776).
AW acknowledges support from the European Commission (STREP ``QCS'' and
Integrated Project ``QESSENCE''), the ERC (Advanced Grant ``IRQUAT''),
a Royal Society Wolfson Merit Award and a Philip Leverhulme Prize.
The Centre for Quantum Technologies is funded by the Singapore
Ministry of Education and the National Research Foundation as part
of the Research Centres of Excellence programme.

\appendix

\section{Appendix: Alternate Proof of Convexity of the Quantum Rate Distortion Function}

\label{app:convexity-QRD}

First recall that the quantum rate
distortion function has the following characterization:%
\[
R^{q}\left(  D\right)  =\lim_{k\rightarrow\infty}\frac{1}{k}\left[
\min_{\mathcal{N}^{\left(  k\right)  }\ :\ \overline{d}(\rho,\mathcal{N}%
^{\left(  k\right)  })\leq D}\ E_{p}\left(  \rho^{\otimes k},\mathcal{N}%
^{\left(  k\right)  }\right)  \right]  .
\]
We would like to show that the expression on the RHS above is convex in $D$:%
\begin{equation}
R^{q}\left(  D_{\lambda}\right)  \leq\lambda R^{q}\left(  D_{1}\right)
+\left(  1-\lambda\right)  R^{q}\left(  D_{2}\right)
,\label{eq:QRD-EoP-convex}%
\end{equation}
where $D_{\lambda}\equiv\lambda D_{1}+\left(  1-\lambda\right)  D_{2}$. We
will choose $k_{1}$ and $k_{2}$ to be large integers such that%
\[
\frac{k_{1}}{k_{1}+k_{2}}\approx\lambda.
\]
(We can actually just take $k_{1}=\left\lceil ak_{2}\right\rceil $ for
$a=\left[  \frac{1}{\lambda}-1\right]  ^{-1}$, so that there is just one
integer to consider.) For $k_{i}$ and $D_{i}$ where $i\in\left\{  1,2\right\}
$, let $\mathcal{N}_{i}^{\left(  k_{i}\right)  }$ be the CPTP\ map that
minimizes%
\[
\min_{\mathcal{N}^{\left(  k_{i}\right)  }\ :\ \overline{d}(\rho
,\mathcal{N}^{\left(  k_{i}\right)  })\leq D_{i}}E_{p}\left(  \rho^{\otimes
k_{i}},\mathcal{N}^{\left(  k_{i}\right)  }\right)  .
\]
Thus, whenever the distortion measure under consideration is linear and
averaged, we have that the distortion caused by the map $\mathcal{N}%
_{1}^{\left(  k_{1}\right)  }\otimes\mathcal{N}_{2}^{\left(  k_{2}\right)  }$
is approximately equal to $D_{\lambda}$. Then we have that%
\begin{align*}
\frac{1}{k_{1}+k_{2}}  & \left[  \min_{\mathcal{N}^{\left(  k_{1}%
+k_{2}\right)  }\ :\ \overline{d}(\rho,\mathcal{N}^{\left(  k_{1}%
+k_{2}\right)  })\leq D}E_{p}\left(  \rho^{\otimes\left(  k_{1}+k_{2}\right)
},\mathcal{N}^{\left(  k_{1}+k_{2}\right)  }\right)  \right]  \\
& \leq\frac{1}{k_{1}+k_{2}}E_{p}\left(  \rho^{\otimes\left(  k_{1}%
+k_{2}\right)  },\mathcal{N}_{1}^{\left(  k_{1}\right)  }\otimes
\mathcal{N}_{2}^{\left(  k_{2}\right)  }\right)  \\
& \leq\frac{1}{k_{1}+k_{2}}\left[  E_{p}\left(  \rho^{\otimes k_{1}%
},\mathcal{N}_{1}^{\left(  k_{1}\right)  }\right)  +E_{p}\left(  \rho^{\otimes
k_{2}},\mathcal{N}_{2}^{\left(  k_{2}\right)  }\right)  \right]  \\
& =\frac{k_{1}}{k_{1}+k_{2}}\left(  \frac{1}{k_{1}}\right)  E_{p}\left(
\rho^{\otimes k_{1}},\mathcal{N}_{1}^{\left(  k_{1}\right)  }\right)
\\
& \,\,\,+\frac{k_{2}}{k_{1}+k_{2}}\left(  \frac{1}{k_{2}}\right)  E_{p}\left(
\rho^{\otimes k_{2}},\mathcal{N}_{2}^{\left(  k_{2}\right)  }\right)  \\
& \approx\lambda\ \frac{1}{k_{1}}E_{p}\left(  \rho^{\otimes k_{1}}%
,\mathcal{N}_{1}^{\left(  k_{1}\right)  }\right)  \\
& \,\,\,+\left(  1-\lambda\right)
\ \frac{1}{k_{2}}E_{p}\left(  \rho^{\otimes k_{2}},\mathcal{N}_{2}^{\left(
k_{2}\right)  }\right).
\end{align*}
The important second inequality follows from subadditivity of the entanglement
of purification. Since the above relation holds for every choice of $k_{2}$
and $k_{1}=\left\lceil ak_{2}\right\rceil $ and the corresponding minimizing
maps $\mathcal{N}_{1}^{\left(  k_{1}\right)  }$ and $\mathcal{N}_{2}^{\left(
k_{2}\right)  }$, it follows that it holds in the limit, implying
(\ref{eq:QRD-EoP-convex}) as desired. Clearly, this argument is very similar
to the observation regarding time-sharing of rate distortion codes in Remark~\ref{rem:convexity}.

\section{Appendix: Converse Proof of Theorem~\ref{thm:EA-QSI-QRD}}

\label{app-ea-converse}
Here we give the details of the proof of the converse part of Theorem~\ref{thm:EA-QSI-QRD}, which deals with entanglement-assisted quantum rate distortion with quantum side information. Continuing from \reff{trace2}, we obtain a 
lower bound on the optimal rate $Q$ of quantum communication needed for this
 \begin{align*}
 2nQ & \equiv 2 \log \left({\rm{dim}}\cH_W\right) \\
 &  \geq2H\left(  W\right) _{\sigma} \\
 &  =H\left(  W\right) _{\sigma}  +H\left(  R^{n}T_{B}B^{n}E\right) _{\sigma} \\
 &  =H\left(  W\right) _{\sigma}  +H\left(  R^{n}T_{B}B^{n}E\right) _{\sigma}  -H\left(
 WR^{n}T_{B}B^{n}E\right) _{\sigma} \\
 &  =I\left(  W;R^{n}T_{B}B^{n}E\right) _{\sigma} \\
 &  \geq I\left(  W;R^{n}T_{B}B^{n}\right)  _{\sigma}\\
 &  =I\left(  WT_{B}B^{n};R^{n}\right)  _{\sigma}+I\left(  W;T_{B}B^{n}\right)
 _{\sigma}-I\left(  R^{n};T_{B}B^{n}\right)  _{\sigma}
   \end{align*}
  \begin{align*}
 &  =I\left(  WT_{B}B^{n};R^{n}\right)  _{\sigma}+I\left(  W;T_{B}B^{n}\right)
 _{\sigma}-I\left(  R^{n};B^{n}\right)  _{\sigma}\\
 &  \geq I\left(  WT_{B}B^{n};R^{n}\right)  _{\sigma}-I\left(  R^{n}%
 ;B^{n}\right)  _{\sigma}\\
 &  \geq I(B^{\prime n}\hat{B}^{n};R^{n})_{\omega}-I(R^{n};\hat{B}^{n}%
 )_{\omega}-n\eps^{\prime}.%
\label{bdd1}
 \end{align*}
In the above, the states $\sigma$ and $\omega$ are as defined above \reff{trace2}.
 The first inequality follows because the entropy of a system is never larger
 than the logarithm of its dimension. The first equality follows because the
 entropy of the marginals of a pure bipartite state are equal. The second
 equality follows because the entropy of a pure state is equal to zero, so that
 $H\left(  WR^{n}T_{B}B^{n}E\right)  =0$. The third equality is an identity.
 The second inequality follows from quantum data processing. The fourth
 equality follows from an identity for the quantum mutual information. The
 fifth equality follows because $R^{n}B^{n}$ and $T_{B}$ are in a product
 state, so that $I\left(  R^{n};T_{B}B^{n}\right)  _{\sigma}=I\left(
 R^{n};B^{n}\right)  _{\sigma}$. The third inequality follows because $I\left(
 W;T_{B}B^{n}\right)  _{\sigma}\geq0$. The last inequality follows from quantum
 data processing and the assumption that that protocol causes only a negligible
 disturbance to the state of the reference and the receiver (the term
 $\eps^{\prime}$ arises from an application of the Alicki-Fannes'
 inequality, where $\eps^{\prime}$ is a function of $\eps$ such that
 $\lim_{\eps\rightarrow0}\eps^{\prime}\left(  \eps\right)  =0$).

Let $B^{\prime}_k$, $k=1,2,\ldots, n$ denote the subsystems (with Hilbert space $\cH_{B^{\prime}}$) constituting the system 
$B^{\prime n}$. Similarly, let ${\hat{B}}_{k}$ and $R_{k}$ ($k=1,2,\ldots, n$) denote the
corresponding subsystems of ${\hat{B}}^n$ and  ${{R}}^n$ respectively. Then
 \begin{align*}
 {\hbox{RHS of \reff{trace2}}} \,&  \geq\sum_{k}\left[  I(B_{k}^{\prime}\hat{B}_{k};R_{k})_{\omega}%
 -I(R_{k};\hat{B}_{k})_{\omega}\right]  -2n\eps^{\prime}\\
 & = \sum_k I(B_k^\prime ;R_k|{\hat{B}}_k)_\omega - 2n \eps^\prime.
\end{align*}

 The first inequality follows from superadditivity of quantum mutual
 information (Lemma~\ref{super}) and from the fact that the
 state on $R^{n}\hat{B}^{n}$ is $\eps$-close in trace distance to a
 tensor-product state (see Lemma~10 of Ref.~\cite{WHBH12}). The second equality
 follows from the identity $I(B_{k}^{\prime}\hat{B}_{k};R_{k})_{\omega}
 -I(R_{k};\hat{B}_{k})_{\omega}=I(B_{k}^{\prime};R_{k}|\hat{B}_{k})_{\omega}$.

At this point, we proceed with an argument similar to that in the converse
proof of Theorem~6 of \cite{DHWW12}. 
Since the state on $R^{n}\hat{B}^{n}$ is not disturbed too much, by Uhlmann's theorem 
we know that there exists 
a CPTP map acting Alice's system alone, such that the conditional mutual information of the state $\omega^{\prime}$ resulting from this map is a lower bound on the 
conditional mutual information $I(B_{k}^{\prime};R_{k}|\hat{B}_{k})_{\omega}$
for every $k\in \{1,2,\ldots, n\}$:%
\[
\geq\sum_{k}I(B_{k}^{\prime};R_{k}|\hat{B}_{k})_{\omega^{\prime}}%
-3n\eps^{\prime}.
\]
Continuing, we have the same set of inequalities as in the last part of the
proof of Theorem~6 of \cite{DHWW12} (with $R_{ea,qsi}$
replacing $R^{qc}_{qsi}$):%
\begin{align*}
&  \geq\sum_{k}R^{q}_{ea,qsi}\left(  d\left(  \rho,\mathcal{F}_{n}^{\left(
k\right)  }\right)  \right)  -3n\eps^{\prime}\\
&  =n\sum_{k}\frac{1}{n}R^{q}_{ea,qsi}\left(  d\left(  \rho,\mathcal{F}%
_{n}^{\left(  k\right)  }\right)  \right)  -3n\eps^{\prime}\\
&  \geq nR^{q}_{ea,qsi}\left(  \sum_{k}\frac{1}{n}d\left(  \rho,\mathcal{F}%
_{n}^{\left(  k\right)  }\right)  \right)  -3n\eps^{\prime}\\
&  \geq nR^{q}_{ea,qsi}\left(  D\right)  -3n\eps^{\prime}.
\end{align*}

\bibliographystyle{IEEEtran}
\bibliography{Ref}

\end{document}